\documentclass[letterpaper,11pt]{article}
\usepackage{fullpage}
\usepackage{amsmath,amsthm,amssymb,color}
\usepackage{thmtools,thm-restate}
\usepackage{hyperref}
\usepackage{graphics}
\usepackage{multirow}
\usepackage{totcount}
\usepackage[utf8]{inputenc}
\usepackage[T1]{fontenc}

        \usepackage{times}
\usepackage{microtype}

\usepackage{mdframed}
\global\mdfdefinestyle{myframe}{leftmargin=.75in,rightmargin=.75in,linecolor=black,linewidth=1.5pt,innertopmargin=10pt,innerbottommargin=10pt}
\global\mdfdefinestyle{eventframe}{leftmargin=1.05in,rightmargin=1.05in,linecolor=black,linewidth=0.75pt,innertopmargin=5pt,innerbottommargin=5pt}

\usepackage[ruled,vlined,linesnumbered,noalgohanging,nofillcomment]{algorithm2e}
\SetAlgoHangIndent{0pt}
\DontPrintSemicolon
	
\newtheorem{theorem}{Theorem}[section]
\newtheorem{lemma}[theorem]{Lemma}

\newtheorem{corollary}[theorem]{Corollary}
\newtheorem{claim}[theorem]{Claim}

\DeclareMathOperator{\polylog}{polylog}

\DeclareMathOperator{\Exp}{{\mathbb E}}

\newcommand{\eqdef}{\stackrel{\text{\tiny\rm def}}{=}}

\newcommand{\algfontsize}{\small}

\newcommand{\R}{\mathbb R}

\newcommand{\eps}{\epsilon}

\newcommand{\CC}{{\mathcal C}}
\newcommand{\DD}{{\mathcal D}}
\newcommand{\EE}{{\mathcal E}}

\newcommand{\oxygen}[1]{\vphantom{\raisebox{3pt}{\mbox{#1}}\raisebox{-2pt}{\mbox{#1}}}#1}

\newcommand{\globalAlgorithm}{\mbox{\tt GlobalAlg}}
\newcommand{\matchHeavy}{\mbox{\tt MatchHeavy}}
\newcommand{\parallelAlgorithm}{\mbox{\tt ParallelAlg}}
\newcommand{\globalPhase}{\mbox{\tt EmulatePhase}}
\newcommand{\localPhase}{\mbox{\tt LocalPhase}}
\newcommand{\parAlg}{\parallelAlgorithm}

\newsavebox{\algDescBox}
\savebox{\algDescBox}{\phantom{\bf Algorithm 5:}}

\newcommand{\algdesc}[1]{\newline \leavevmode\usebox{\algDescBox}\,\,#1}

\newcommand{\istar}{{i_{\star}}}%

\newcommand{\threshold}{\Delta}
\newcommand{\phases}{\tau}
\newcommand{\machines}{m}
\newcommand{\prR}{\mu_R}
\newcommand{\prH}{\mu_H}
\newcommand{\prHmult}{\alpha}
\newcommand{\prF}{\mu_F}

\newcommand{\baseR}{\chi}
\newcommand{\baseH}{\psi}

\newcommand{\undefine}[1]{\let#1\ThisCommandWasUndefined}

\usepackage{stmaryrd}
\newcommand{\range}[1]{\left\llbracket{#1}\right\rrbracket}

\newcommand{\degEst}[1]{\widehat{d}_{#1}}

\newtotcounter{EventsCounter}
\newcommand{\event}[1]{Event~\ref{event:#1}}

\newenvironment{NewEvent}[1]{\vspace*{2mm}\newline\begin{minipage}{\textwidth}\begin{mdframed}[style=eventframe]\begin{center}\refstepcounter{EventsCounter}\label{event:#1}{\bf \quad Event \theEventsCounter}\end{center}
\setlength{\belowdisplayskip}{0pt}\setlength{\belowdisplayshortskip}{0pt}\setlength{\abovedisplayskip}{3pt}\setlength{\abovedisplayshortskip}{3pt}
\vspace{5pt}\noindent}%
{\end{mdframed}\end{minipage}\vspace*{2mm}}

\newcommand{\allEvents}{Events~1--\total{EventsCounter}}

\title{Round Compression for Parallel Matching Algorithms%
}

\author{Artur Czumaj \\ University of Warwick \\a.czumaj@warwick.ac.uk
\and
Jakub Łącki \\ Google Research, New York \\ jlacki@google.com
\and
Aleksander Mądry \\ MIT \\ madry@mit.edu
\and\\
Slobodan Mitrovi\'c \\ EPFL \\ slobodan.mitrovic@epfl.ch
\and\\
Krzysztof Onak \\ IBM Research \\ konak@us.ibm.com
\and\\
Piotr Sankowski \\ University of Warsaw \\ sank@mimuw.edu.pl}

\date{November 2017}
\date{}

\begin{document}

\begin{titlepage}
\maketitle
\thispagestyle{empty}
\begin{abstract}
For over a decade now we have been witnessing the success of \emph{massive parallel computation} (MPC) frameworks, such as MapReduce, Hadoop, Dryad, or Spark. One of the reasons for their success is the fact that these frameworks are able to accurately capture the nature of large-scale computation. In particular, compared to the classic distributed algorithms or PRAM models, these frameworks allow for much more local computation. The fundamental question that arises in this context is though: can we leverage this additional power to obtain even faster parallel algorithms?

A prominent example here is the \emph{maximum matching} problem---one of the most classic graph problems.
It is well known that in the PRAM model one can compute a 2-approximate maximum matching in $O(\log{n})$ rounds. However, the exact complexity of this problem in the MPC framework is still far from understood. Lattanzi et al.~(SPAA 2011) showed that if each machine has $n^{1+\Omega(1)}$ memory, this problem can also be solved $2$-approximately in a constant number of rounds. These techniques, as well as the approaches developed in the follow up work, seem though to get stuck in a fundamental way at roughly $O(\log{n})$ rounds once we enter the (at most) near-linear memory regime. It is thus entirely possible that in this regime, which captures in particular the case of sparse graph computations, the best MPC round complexity matches what one can already get in the PRAM model, without the need to take advantage of the extra local computation power.

In this paper, we finally refute that perplexing possibility. That is, we break the above $O(\log n)$ round complexity bound even in the case of \emph{slightly sublinear} memory per machine. In fact, our improvement here is {\em almost exponential}: we are able to deliver a $(2+\epsilon)$-approximation to maximum matching, for any fixed constant $\epsilon>0$, in $O\left((\log \log n)^2\right)$ rounds.

To establish our result we need to deviate from the previous work in two important ways that are crucial for exploiting the power of the MPC model, as compared to the PRAM model. Firstly, we use \emph{vertex--based} graph partitioning, instead of the edge--based approaches that were utilized so far. Secondly, we develop a technique of \emph{round compression}. This technique enables one to take a (distributed) algorithm that computes an $O(1)$-approximation of maximum matching in $O(\log n)$ \emph{independent} PRAM phases and implement a super-constant number of these phases in only a constant number of MPC rounds.
\end{abstract}

\end{titlepage}

\section{Introduction}

Over the last decade, massive parallelism became a major paradigm in computing, and we have witnessed the deployment of a number of very successful massively parallel computation frameworks, such as MapReduce~\cite{dg04,dg08}, Hadoop~\cite{White:2012}, Dryad~\cite{Isard:2007}, or Spark~\cite{ZahariaCFSS10}. This paradigm and the corresponding models of computation are rather different from classical parallel algorithms models considered widely in literature, such as the PRAM model. In particular, in this paper,  we study the {\em Massive Parallel Computation (MPC)} model (also known as Massively Parallel Communication model) that was abstracted out of capabilities of existing systems, starting with the work of Karloff, Suri, and Vassilvitskii~\cite{KarloffSV10,goodrich2011sorting,BeameKS13,AndoniNOY14,BeameKS14}. The main difference between this model and the PRAM model is that the MPC model allows for much more
(in principle, unbounded) local computation. This enables it to capture a more ``coarse--grained,'' and thus, potentially, more meaningful aspect of parallelism.  It is often possible to simulate one clock step of PRAM in a constant number of rounds on MPC~\cite{KarloffSV10,goodrich2011sorting}. This implies that algorithms for the PRAM model usually give rise to MPC algorithms without incurring any asymptotic blow up in the number of parallel rounds.  As a result, a vast body of work on PRAM algorithms naturally translates to the new model.

It is thus natural to wonder: Are the MPC parallel round bounds ``inherited'' from the PRAM model tight? In particular, which problems can be solved in significantly {\em smaller} number of MPC rounds than what the lower bounds established for the PRAM model suggest?

It is not hard to come up with an example of a problem for which indeed the MPC parallel round number is much smaller than its PRAM round complexity. For instance, computing the parity of $n$ Boolean values takes only $O(1)$ parallel rounds in the MPC model when space per machine is $n^{\Omega(1)}$, while on PRAM it provably requires $\Omega(\log n/\log \log n)$ time~\cite{Beame:1987} (as long as the total number of processors is polynomial). However, the answer is typically less obvious for other problems. This is particularly the case for graph problems, whose study in a variant of the MPC model was initiated already by Karloff et al.~\cite{KarloffSV10}.

In this paper, we focus on one such problem, which is also one of the most central graph problems both in sequential and parallel computations: maximum matching. Maximum matchings have
been  the cornerstone of algorithmic research since 1950s and their study inspired many important ideas, including the complexity class~P \cite{Edmonds65paths}.
In the PRAM model we can compute $(1+\epsilon)$-approximate matching in $O(\log n)$ rounds~\cite{Lotker:2015} using randomization. Deterministically, a $(2+\epsilon)$-approximation can be computed in $O\left(\log^2 n\right)$ rounds~\cite{pram2017}. We note that these results hold in a distributed message passing setting, where
 processors are located at graph nodes and can communicate only with neighbors. In such a distributed setting,  $\Omega\left(\sqrt{\log  n / \log{\log{n}}}\right)$ time lower
 bound is known for computing any constant approximation to maximum matching~\cite{Kuhn:2006}.

So far, in the MPC setting, the prior results are due to Lattanzi, Moseley, Suri, and Vassilvitskii~\cite{LattanziMSV11}, Ahn and Guha~\cite{AhnG15} and Assadi and Khanna~\cite{AssadiK17}. Lattanzi et al.~\cite{LattanziMSV11} put forth  algorithms for several graph problems, such as connected components, minimum spanning tree, and maximum matching problem, that were based on a so-called \emph{filtering technique}. In particular, using this technique, they have obtained an algorithm that can compute a $2$-approximation to maximum matching in $O(1/\delta)$ MPC rounds, provided $S$, the space per machine, is significantly larger than the total number of vertices $n$, that is $S = \Omega\left(n^{1+\delta}\right)$, for any constant $\delta \in (0,1)$. Later on, Ahn and Guha~\cite{AhnG15} provided an improved algorithm that computes a $(1+\eps)$-approximation in $O(1/(\delta \eps))$ rounds, provided $S=\Omega\left(n^{1+\delta}\right)$, for any constant $\delta>0$. Both these results, however,  crucially require that space per machine is significantly superlinear in $n$, the number of vertices.
In fact, if the space $S$ is linear in $n$, which is a very natural setting for massively parallel graph algorithms, the performance of both these algorithms degrades to $O(\log n)$ parallel rounds, which matches what was known for the PRAM model.
Recently, Assadi and Khanna~\cite{AssadiK17} showed how to construct randomized composable coresets of size $\tilde{O}(n)$ that give an $O(1)$-approximation for maximum matching. Their techniques apply to the MPC model only if the space per machine is $\tilde{O}(n \sqrt{n})$.

We also note that the known PRAM maximal independent set and maximal matching algorithms~\cite{Luby86,AlonBI86,II86} can be used to find a maximal matching (i.e., 2-approximation to maximum matching) in $O(\log n)$ MPC rounds as long as space per machine is at least $n^{\Omega(1)}$ (i.e., $S \ge n^{c}$ for some constant $c > 0$). We omit further details here, except mentioning that a more or less direct simulation of those algorithms is possible via an $O(1)$-round sorting subroutine~\cite{goodrich2011sorting}.

The above results give rise to the following fundamental question: Can the maximum matching be (approximately) solved in $o(\log n)$ parallel rounds in $O(n)$ space per machine?
The main result of this paper is an affirmative answer to that question. We show that, for any $S = \Omega(n)$, one can obtain an $O(1)$-approximation to maximum matching using $O\left((\log \log n)^2\right)$ parallel MPC rounds. So, not only do we break the existing $\Omega(\log n)$ barrier, but also provide an exponential improvement over the previous work. Our algorithm can also provide a $(2+\eps)$, instead of $O(1)$-approximation, at the expense of the number of parallel rounds increasing by a factor of $O(\log (1/\eps))$. Finally, our approach can also provide algorithms that have $o(\log n)$ parallel round complexity also in the regime of $S$ being (mildly) sublinear.
For instance, we obtain $O\left((\log \log n)^2\right)$ MPC rounds even if space per machine is $S = n / (\log n)^{O(\log \log n)}$.
The exact comparison of our bounds with previous results is given in Table~\ref{table-results}.

\subsection{The model}
In this work, we adopt a version of the model introduced by Karloff, Suri, and Vassilvitskii~\cite{KarloffSV10} and refined in later works~\cite{goodrich2011sorting,BeameKS13,AndoniNOY14}. We call it \emph{massive parallel computation} (MPC), which is a mutation of the name proposed by Beame et al.~\cite{BeameKS13}.

In the MPC model, we have $\machines$ machines at our disposal and each of them has $S$ words of space.
Initially, each machine receives its share of the input. In our case, the input is a collection $E$ of edges and each machine receives approximately $|E|/\machines$ of them.

The computation proceeds in \emph{rounds}.
During the round, each of the machines processes its local data without communicating with other machines.
At the end of each round, machines exchange messages.
Each message is sent only to a single machine specified by the machine that is sending the message.
All messages sent and received by each machine in each round have to fit into the machine's local memory. Hence, their total length is bounded by $S$.\footnote{This for instance allows a machine to send a single word to $S/100$ machines or $S/100$ words to one machine, but not $S/100$ words to $S/100$ machines if $S = \omega(1)$, even if the messages are identical.}
This in particular implies that the total communication of the MPC model is bounded by $\machines \cdot S$ in each round.
The messages are processed by recipients in the next round.

At the end of the computation, machines collectively output the solution. The data output by each machine has to fit in its local memory. Hence again, each machine can output at most $S$ words.

\paragraph{The range of values for $S$ and $\machines$.} If the input is of size $N$, one usually wants $S$ sublinear in the $N$, and the total space across all the machines to be at least $N$---so the input fits onto the machines---and ideally not much larger. Formally, one usually considers $S \in \Theta\left(N^{1 - \epsilon}\right)$, for some $\epsilon > 0$.

In this paper, the focus is on graph algorithms. If $n$ is the number of vertices in the graph, the input size can be as large as $\Theta\left(n^2\right)$. Our parallel algorithm requires $\Theta(n)$ space per machine (or even slightly less), which is polynomially less than the size of the input for dense graphs.

\paragraph{Sparse graphs.}
Many practical large graphs are believed to have only $O(n)$ edges. One natural example is social networks, in which most participants are likely to have a bounded number of friends. The additional advantage of our approach is that it allows for a small number of processing rounds even if a sparse input graph does not fit onto a single machine. If a small number---say, $f(n)$---of machines is needed even to store the graph, our algorithm still requires only $O\left((\log \log n)^2 + \log f(n)\right)$ rounds for $O(n/f(n))$ space per machine.

\paragraph{Communication vs.\ computation complexity.} The main focus of this work is the number of (communication) rounds required to finish computation. Also, even though we do not make an effort to explicitly bound it, it is apparent from the design of our algorithms that every machines performs $O(S \polylog{S})$ computation steps locally. This in particular implies that the overall work across all the machines is $O(r N \polylog{S})$, where $r$ is the number of rounds and $N$ is the input size (i.e., the number of edges).

The total communication during the computation is $O(r N)$ words. This is at most $O\left(rn^2\right)$ words and it is known that computing a $(1+\eps)$-approximate matching in the message passing model with $\Theta(n)$ edges per player may require $\Omega\left(n^2/(1+\eps)^2\right)$ bits of~communication~\cite{HuangRVZ15}. Since our value of $r$ is $O\left((\log \log n)^2\right)$ when $\Theta(n)$ edges are assigned to each player, we lose a factor of $\tilde\Theta(\log n)$ compared to this lower bound if words (and vertex identifiers) have $\Theta(\log n)$ bits.

\subsection{Our results}
In our work, we focus on computing an $O(1)$-approximate maximum matching in the MPC model. We collect our results and compare to the previous work in Table~\ref{table-results}.
\begin{table}[h!]\algfontsize
\newcommand{\leftPart}{(\log \log n)^2 + \log f(n)}%
\newcommand{\rightPart}{\log(1/\eps)}%
\centering
\begin{tabular}{|c|c|c|c|c|}
 \hline
 \bf \oxygen{Source} & \bf Approx. & \bf Space & \bf Rounds & \bf Remarks \\
 \hline
 \multirow{2}{*}{\cite{LattanziMSV11}} & \multirow{2}{*}{2} & \oxygen{$n^{1+\Omega(1)}$} & \oxygen{$O(1)$} & \multirow{2}{*}{Maximal matching} \\
 \cline{3-4}
 & & \oxygen{$O(n)$} & \oxygen{$O(\log n)$} & \\
 \hline
 \cite{AhnG15} & $1+\eps$ & \oxygen{$O\left(n^{1+1/p}\right)$} & \oxygen{$O(p/\eps)$} & $p > 1$ \\
 \hline
 & \multirow{2}{*}{2} & \multirow{2}{*}{\oxygen{$n^{\Omega(1)}$}} & \multirow{2}{*}{\oxygen{$O(\log n)$}} & Maximal matching \\
 &                    &                                           &                                       & Simulate~\cite{Luby86,AlonBI86,II86} \\
 \hline
 & $O(1)$ & \multirow{2}{*}{\oxygen{$O(n)$}} & \oxygen{$O\left((\log \log n)^2\right)$} &\\
 \cline{2-2}\cline{4-4}
 & $2+\eps$ & & \oxygen{$O\left((\log \log n)^2 \cdot \rightPart\right)$} & \oxygen{$\eps \in (0,1/2)$}\\
 \cline{2-4}
 here & $O(1)$ & \multirow{2}{*}{$O(n)/f(n)$} & \oxygen{$O\left(\leftPart\right)$} & $2 \le f(n) = O\left(n^{1/2}\right)$
 \\
 \cline{2-2}\cline{4-4}
 & $2+\eps$ & & \oxygen{$O\left((\log \log n)^2 + \log f(n)\right)\cdot\rightPart$} & \\
 \hline
\end{tabular}
\caption{Comparison of our results for computing approximate maximum size matchings to the previous results for the MPC model.}
\label{table-results}
\end{table}
The table presents two interesting regimes for our algorithms. On the one hand, when the space per machine is $S=O(n)$, we obtain an algorithm that
requires $O((\log \log n)^2)$ rounds. This is the first known algorithm that, with linear space per machine, breaks
 the $O(\log n)$ round barrier.  On the other hand, in the mildly sublinear regime of space per machine, i.e., when $S=O(n/f(n))$, for some function $f(n)$ that is $n^{o(1)}$, we obtain an algorithm that still requires $o(\log n)$ rounds. This, again is the first such result in this regime. In particular, we prove the following result.
\begin{theorem}\label{theorem-main}
	 There exists an MPC algorithm that constructs an $O(1)$-approximation to maximum matching with constant probability in $O\left((\log \log n)^2 + \max{\left(\log{\tfrac{n}{S}}, 0 \right)} \right)$ rounds, where $S=n^{\Omega(1)}$ is the amount of space on each machine.
\end{theorem}
As a corollary, we obtain the following result that provides nearly $2$-approximate maximum matching.
\begin{corollary}\label{corollary-main}
         For any $\eps \in (0,\frac{1}{2})$, there exists an MPC algorithm that constructs a $(2+\eps)$-approximation to maximum matching with $99/100$ probability in $O\left((\log \log n)^2 + \max{\left(\log{\tfrac{n}{S}}, 0 \right)} \right) \cdot \log(1/\eps)$ rounds, where $S=n^{\Omega(1)}$ is the amount of space on each machine.
\end{corollary}

Assadi et al.~\cite{ABBMS17} observe that one can use a technique of McGregor~\cite{McGregor05} to extend the algorithm to compute a $(1+\eps)$-approximation in 
$O((\log \log n)^2)\cdot (1/\eps)^{O(1/\eps)}$ rounds.

It should also be noted that (as pointed out to us by Seth Pettie) any $O(1)$-approximation algorithm for unweighted matchings can be used to obtain a $(2+\eps)$-approximation algorithm for weighted matchings (see Section 4 of his paper with Lotker and Patt-Shamir \cite{Lotker:2015} for details). In our setting this implies that Theorem \ref{theorem-main} yields an algorithm that computes a \emph{$(2+\eps)$-approximation} to maximum weight matching in $O((\log \log n)^2\cdot (1/\eps))$ rounds and $O(n \log n)$ space per machine.

\subsection{Related work}

We note that there were efforts at modeling MapReduce computation~\cite{FeldmanMSSS10} before the work of Karloff et al. Also a recent work \cite{RoughgardenVW16} investigates the complexity of the MPC model.

In the \emph{filtering} technique, introduced by Lattanzi et al.~\cite{LattanziMSV11}, the input graph is iteratively sparsified until it can be stored on a single machine. For the matching problem, the sparsification is achieved by first obtaining a small sample of edges, then finding a maximal matching in the sample, and finally removing all the matched vertices. Once a sufficiently small graph is obtained, a maximal matching is computed on a single machine. In the $S = \Theta(n)$ regime, the authors show that their approach reduces the number of edges by a constant factor in each iteration. Despite this guarantee, until the very last step, each iteration may make little progress towards obtaining even an approximate maximal matching, resulting in a $O(\log{n})$ round complexity of the algorithm. Similarly, the results of Ahn and Guha~\cite{AhnG15} require $n^{1 + \Omega(1)}$ space per machine to compute a $O(1)$-approximate maximum weight matching in a constant number of rounds and do not imply a similar bound for the case of linear space.

We note that the algorithm of Lattanzi et al.~\cite{LattanziMSV11} cannot be turned easily into a fast approximation algorithm when space per machine is sublinear. Even with $\Theta(n)$ space, their method is able to remove only a constant fraction of edges from the graph in each iteration, so $\Omega(\log n)$ rounds are needed until only a matching is left. When $S=\Theta(n)$, their algorithm
works as follows: sample uniformly at random $\Theta(n)$ edges of the graph, find maximal matching on the sampled set, remove the matched vertices, and repeat. We do not provide a formal proof here, but on the following graph this algorithm requires $\tilde\Omega(\log n)$ rounds, even to discover a constant factor approximation. Consider a graph consisting of $t$ separate regular graphs of degree $2^i$, for $0\le i \le t-1$, each on $2^t$ vertices. This graph has $t2^t$ nodes and the algorithm requires $\tilde\Omega(t)$ rounds even to find a constant approximate matching. The algorithm chooses edges uniformly at random, and few edges are selected each round from all but the densest remaining subgraphs. Thus, it takes multiple rounds until a matching of significant size is constructed for sparser subgraphs.
This example emphasizes the weakness of direct edge sampling and motivates our vertex sampling scheme that we introduce in this paper.

Similarly, Ahn and Gupta~\cite{AhnG15} build on the filtering approach of Lattanzi et al.\ and design a primal-dual method for computing a $(1 + \epsilon)$-approximate weighted maximum matching. They show that each iteration of their distributed algorithm either makes large progress in the dual, or they can construct a large approximate matching. Regardless of their new insights, their approach is inherently edge-sampling based and does not break the $O(\log{n})$ round complexity barrier when $S = O(n)$.

Despite the fact that MPC model is rather new, computing matching is an important problem in this model, as the above mentioned two papers demonstrate.
This is further witnessed by the fact that the distributed and parallel complexity of maximal matching has been studied for many years already.
The best deterministic PRAM maximal matching algorithm, due to Israeli and Shiloach~\cite{IsraeliS86}, runs in $O\left(\log^3 n\right)$ rounds.
Israeli and Itai~\cite{II86} gave a randomized algorithm for this problem that runs in $O(\log n)$ rounds.
Their algorithm works as well in CONGEST, a distributed message-passing model with a processor assigned to each vertex and a limit on the amount of information sent along each edge per round. A more recent paper by Lotker, Patt-Shamir, and Pettie~\cite{Lotker:2015} gives a $(1+\epsilon)$-approximation to maximum matching in $O(\log n)$ rounds also in the CONGEST model, for any constant $\epsilon > 0$.
On the deterministic front, in the LOCAL model, which is a relaxation of CONGEST that allows for an arbitrary amount of data sent along each edge,
a line of research initiated by Hańćkowiak, Karoński, and Panconesi~\cite{hkp1,hkp2} led to an $O\left(\log^3 n\right)$-round algorithm by Fischer and Ghaffari~\cite{pram2017}.

On the negative side, Kuhn, Moscibroda, and Wattenhofer~\cite{Kuhn:2006} showed that any distributed algorithm, randomized or
deterministic, when communication is only between neighbors requires $\Omega\left(\sqrt{\log  n / \log{\log{n}}}\right)$ rounds
to compute a constant approximation to maximum matching. This lower bound applies to
all distributed algorithms that have been mentioned above. Our algorithm circumvents this lower bound by
loosening the only possible assumption there is to be loosened: single-hop communication. In a sense, we assign subgraphs to multiple machines
and allow multi-hop communication between nodes in each subgraph.

Finally, the ideas behind the peeling algorithm that is a starting point for this paper can be traced back to the papers of Israeli, Itai, and Shiloach~\cite{II86,IsraeliS86}, which can be interpreted as matching high-degree vertices first in order to reduce the maximum degree.
A sample distributed algorithm given in a work of Parnas and Ron~\cite{ParnasR07} uses this idea to compute an $O(\log{n})$ approximation for vertex cover. Their algorithm was extended by Onak and Rubinfeld~\cite{OR10} in order to provide an $O(1)$-approximation for vertex cover and maximum matching in a dynamic version of the problems. This was achieved by randomly matching high-degree vertices to their neighbors in consecutive phases while reducing the maximum degree in the remaining graph. This approach was further developed in the dynamic graph setting by a number of papers~\cite{BhattacharyaHI15,BhattacharyaHN16,BhattacharyaHN17,BhattacharyaCH17}. Ideas similar to those in the paper of Parnas and Ron~\cite{ParnasR07} were also used to compute polylogarithmic approximation in the streaming model by Kapralov, Khanna, and Sudan~\cite{KapralovKS14}. Our version of the peeling algorithm was directly inspired by the work of Onak and Rubinfeld~\cite{OR10} and features important modifications in order to make our analysis go through.

\subsection{Future challenges}

We show a parallel matching algorithm in the MPC model by taking an algorithm that can be seen as a distributed algorithm in the so-called
LOCAL model. This algorithm requires $\Theta(\log n)$ rounds and can be simulated in $\Theta(\log n)$ MPC rounds relatively easily with $n^{\Omega(1)}$ space per machine. We develop an approximate version of the algorithm that uses much fewer rounds by repeatedly compressing a superconstant number of rounds of the original algorithm to $O(1)$ rounds. It is a great question if this kind of speedup can be obtained for other---either distributed or PRAM---algorithms.

As for the specific problem considered in this paper, an obvious question is whether our round complexity is optimal. We conjecture that there is a better algorithm that requires $O(\log \log n)$ rounds, the square root of our complexity. Unfortunately, a factor of $\log n$ in one of our functions (see the logarithmic factor in $\prHmult$, a parameter defined later in the paper) propagates to the round complexity, where it imposes a penalty of $\log \log n$.

Note also that as opposed to the paper of Onak and Rubinfeld~\cite{OR10}, we do not obtain an $O(1)$-approximation to vertex cover. This stems from the fact that we discard so-called reference sets, which can be much bigger than the minimum vertex cover. This is unfortunately necessary in our analysis. Is there a way to fix this shortcoming of our approach?

Finally, we suspect that there is a simpler algorithm for the problem that avoids the intricacies of our approach and proceeds by simply greedily matching high-degree vertices on induced subgraphs without sophisticated sampling in every phase. Unfortunately, we do not know how to analyze this kind of approach.

\subsection{Recent developments}
Since an earlier version of this work was shared on arXiv, it has inspired two followup works. First, Assadi~\cite{Assadi17} applied the round compression idea to the distributed $O(\log n)$-approximation algorithm for vertex cover of Parnas and Ron~\cite{ParnasR07}. Using techniques from his recent work with Khanna~\cite{AssadiK17}, he gave a simple MPC algorithm that in $O(\log \log n)$ rounds and $n/\polylog(n)$ space per machine computes an $O(\log n)$-approximation to minimum vertex cover.

Second, a new paper by Assadi et al.~\cite{ABBMS17} addresses, among other things, several open questions from this paper.
They give an MPC algorithm that computes $O(1)$-approximation to both vertex cover and maximum matching in $O(\log \log n)$ rounds and $\tilde O(n)$ space per machine (though the space is strictly superlinear). Their result builds on techniques developed originally for dynamic matching algorithms~\cite{BS15,BS16} and composable coresets~\cite{AssadiK17}. It is worth to note that their construction critically relies on the vertex sampling approach (i.e., random assignment of vertices to machines) introduced in our work.

\subsection{Notation}
For a graph $G = (V,E)$ and $V' \subseteq V$, we write $G[V']$ to denote the subgraph of $G$ induced by $V'$. Formally, $G[V'] \eqdef \left( V', E \cap (V' \times V') \right)$.
We also write $N(v)$ to denote the set of neighbors of a vertex $v$ in $G$.

\section{Overview}
In this section we present the main ideas and techniques behind our result.
Our paper contains two main technical contributions.

First, our algorithm \emph{randomly partitions vertices} across the machines, and on each machine considers only the corresponding induced graph.
We prove that it suffices to consider these induced subgraphs to obtain an approximate maximum matching.
Note that this approach greatly deviates from previous works, that used edge based partitioning.

Second, we introduce a \emph{round compression} technique.
Namely, we start with an algorithm that executes $O(\log n)$ phases and can be naturally implemented in $O(\log n)$ MPC rounds and then demonstrate how to emulate this algorithm using only $o(\log n)$ MPC rounds.
The underlying idea is quite simple: each machine independently runs multiple phases of the initial algorithm.
This approach, however, has obvious challenges since the machines cannot communicate in a single round of the MPC algorithm.
The rest of the section is devoted to describing our approach and illustrating how to overcome these challenges.



\subsection{Vertex based sampling}
The algorithms for computing maximal matching in PRAM and their simulations in the MPC model \cite{Luby86, AlonBI86, IsraeliS86, II86} are designed to, roughly speaking, either halve the number of the edges or halve the maximum degree in each round. Therefore, in the worst case those algorithms inherently require $\Omega(\log{n})$  rounds to compute a maximal matching.

On the other hand, all the algorithm for the maximal matching problem in the MPC model prior to ours (\cite{LattanziMSV11, AhnG15, AssadiK17}) process the input graph by discarding edges, and eventually aggregate the remaining edges on a single machine to decide which of them are part of the final matching.
It is not known how to design approaches similar to \cite{LattanziMSV11, AhnG15, AssadiK17} while avoiding a step in which the maximal matching computation is performed on a single machine. This seems to be a barrier for improving upon $O(\log{n})$ rounds, if the space available on each machine is $O(n)$.




The starting point of our new approach is alleviating this issue by resorting to a more careful vertex based sampling. Specifically, at each round, we randomly partition the vertex set into vertex sets $V_1, \ldots , V_m$ and consider induced graphs on those subsets independently. Such sampling scheme has the following handy property: the union of matchings obtained across the machines is still a matching. Furthermore, we show that for the appropriate setting of parameters this sampling scheme allows us to handle vertices of a wide range of degrees in a single round, unlike handling only high-degree vertices (that is, vertices with degree within a constant factor of the maximum degree) as guaranteed by \cite{II86, IsraeliS86}.

\subsection{Global algorithm}

To design an algorithm executed on machines locally, we start from a sequential peeling algorithm \globalAlgorithm{} (see Algorithm~\ref{alg:global}), which is a modified version of an algorithm used by Onak and Rubinfeld~\cite{OR10}. The algorithm had to be significantly adjusted in order to make our later analysis of a parallel version possible.

The execution of $\globalAlgorithm$ is divided into $\Theta(\log n)$ \emph{phases}. In each phase, the algorithm first computes a set $H$ of high-degree vertices. Then it selects a set $F$ of vertices, which we call \emph{friends}.
Next the algorithm selects a matching $\widetilde{M}$ between $H$ and $F$, using a simple randomized strategy. $F$ is carefully constructed so that both $F$ and $\widetilde{M}$ are likely to be of order $\Theta(|H|)$.
Finally, the algorithm removes all vertices in $H\cup F$, hence reducing the maximum vertex degree in the graph by a constant factor, and proceeds to the next phase. The central property of $\globalAlgorithm$ is that it returns an $O(1)$ approximation to maximum matching with constant probability (Corollary~\ref{cor:global_alg_is_good}).
A detailed discussion of \globalAlgorithm{} is given in Section~\ref{sec:global-algorithm}.

\begin{algorithm}[ht]\algfontsize
  \KwIn{Graph $G = (V,E)$ of maximum degree at most $\widetilde\threshold$}
  \KwOut{A matching in $G$}
  \BlankLine
  \caption{\globalAlgorithm($G, \widetilde \threshold$) \algdesc Global matching algorithm}
  \label{alg:global}
  $\threshold \leftarrow \widetilde\threshold$, $M \leftarrow \emptyset$, $V' \leftarrow V$\\
  \While{$\threshold \ge 1$\label{line:global:main_loop_start}}{

   \tcc{\small Invariant:~the maximum degree in $G[V']$ is at most $\threshold$}

   Let $H\subset V'$ be a set of vertices of degree at least $\threshold/2$ in $G[V']$. We call vertices in $H$ \emph{heavy}.\label{line:global:high-degree}

   Create a set $F$ of \emph{friends} by selecting each vertex $v \in V'$ independently with probability $|N(v) \cap H|/4\threshold$.\label{line:global:friends}

   Compute a matching $\widetilde M$ in $G[H \cup F]$ using $\matchHeavy(H,F)$ and add it to $M$.\label{line:compute_matching}

   $V' \leftarrow V' \setminus (H \cup F)$, $\threshold \leftarrow \threshold/2$\label{line:global:main_loop_end}

  }
\Return{$M$}
\end{algorithm}

\begin{algorithm}[ht]\label{alg:randomized_matching}\algfontsize
  \KwIn{set $H$ of heavy vertices and set $F$ of friends}
	\KwOut{a matching in $G[H \cup F]$}
  \BlankLine
	\caption{\matchHeavy$(H, F)$\algdesc{Computing a matching in $G[H \cup F]$}}
   For every vertex $v \in F$ pick uniformly at random a heavy neighbor $v_\star$ in $N(v) \cap H$.\label{step:select}\label{line:randomized_matching:v_star}

Independently at random color each vertex in $H \cup F$ either red or blue.

Select the following subset of edges: $E_\star \leftarrow \{(v,v_\star): v \in F \land \text{$v$ is red} \land v_\star \in H \land \text{$v_\star$ is blue}\}$.

For every blue vertex $w$ incident to an edge in $E_\star$, select one such edge and add it to $\widetilde M$.\label{step:match}

\Return{$\widetilde M$}
\end{algorithm}

\subsection{Parallel emulation of the global algorithm (Section~\ref{sec:emulate-phase})}
The following two ways could be used to execute $\globalAlgorithm$ in the MPC model: (1) place the whole graph on one machine, and trivially execute all the phases of $\globalAlgorithm$ in a single round; or (2) simulate one phase of $\globalAlgorithm$ in one MPC round while using $O(n)$ space per machine, by distributing vertices randomly onto machines (see Section~\ref{sec:implementation-of-centralized-algorithm} for details). However, each of these approaches has severe drawbacks. The first approach requires $\Theta(|E|)$ space per machine, which is likely to be prohibitive for large graphs. On the other hand, while the second approach uses $O(n)$ space, it requires $\Theta(\log{n})$ rounds of MPC computation. We achieve the best of both worlds by showing how to emulate the behavior
of \emph{multiple phases} of $\globalAlgorithm$ in a \emph{single MPC round} with each machine using $O(n)$ space, thus obtaining an MPC algorithm requiring $o(\log n)$ rounds. More specifically, we show that it is possible to emulate the behavior of $\globalAlgorithm$ in $O\left((\log \log{n})^2\right)$ rounds with each machine using $O(n)$ (or even only $n/(\log n)^{O(\log \log n)}$) space.

Before we provide more details about our parallel multi-phase emulation of $\globalAlgorithm$, let us mention the main obstacle such an emulation encounters.
At the beginning of every phase, $\globalAlgorithm$ has access to the full graph. Therefore, it can easily compute the set of heavy vertices $H$. On the other hand, machines in our MPC algorithm use $O(n)$ space and thus have access only to a small subgraph of the input graph (when $|E| \gg n$). In the first phase this is not a big issue, as, thanks to randomness, each machine can estimate the degrees of high-degree vertices. However, the degrees of vertices can significantly change from phase to phase. Therefore, after each phase it is not clear how to select high-degree vertices in the next phase without inspecting the entire graph again. Hence, one of the main challenges in designing a multi-phase emulation of $\globalAlgorithm$ is to ensure that machines at the beginning of every phase can estimate \emph{global} degrees of vertices well enough to identify the set of heavy vertices, while each machine still having access only to its local subgraph. This property is achieved using a few modifications to the algorithm.

\subsubsection{Preserving randomness}
Our algorithm partitions the vertex set into $\machines$ disjoint subsets $V_i$ by assigning each vertex independently and uniformly at random.
Then the graph induced by each subset $V_i$ is processed on a separate machine. Each machine finds a set of heavy vertices, $H_i$, by estimating the global degree of each vertex of $V_i$.
It is not hard to argue (using a standard concentration bound) that there is enough randomness in the initial partition so that local degrees in each induced subgraph roughly correspond to the global degrees. Hence,
after the described partitioning, sets $H$ and $\bigcup_{i \in [\machines]} H_i$ have very similar properties.
This observation crucially relies on the fact that initially the vertices are distributed \emph{independently} and \emph{uniformly} at random.

However, if one attempts to execute the second phase of $\globalAlgorithm$ without randomly reassigning vertices to sets after the first phase, the remaining vertices are no longer distributed independently and uniformly at random.
In other words, after inspecting the neighborhood of every vertex locally and making a decision based on it, the randomness of the initial random partition may significantly decrease.

Let us now make the following thought experiment. Imagine for a moment that there is an algorithm that emulates multiples phases of $\globalAlgorithm$ in parallel and in every phase inspects \textbf{only} the vertices that end-up being matched. Then, from the point of view of the algorithm, the vertices that are not matched so far are still distributed independently and uniformly at random across the machines. Or, saying in a different way, if randomness of some vertices is not inspected while emulating a phase, then at the beginning of the next phase those vertices still have the same distribution as in the beginning of that MPC round. But, how does an algorithm learn about vertices that should be matched by inspecting no other vertex? How does the algorithm learn even only about high-degree vertices without looking at their neighborhood?\\
In the sequel we show how to design an algorithm that looks only "slightly" at the vertices that do not end-up being matched. As we prove, that is sufficient to design a multi-phase emulation of $\globalAlgorithm$.

We now discuss in more detail how to preserve two crucial properties of our vertex assignments throughout the execution of multiple phases: independent and nearly-uniform distribution.

\subsubsection{Independence (Lemma~\ref{lem:independence})}
As noted above, it is not clear how to compute vertex degrees without inspecting their local neighborhood.
A key, and at first sight counter-intuitive, step in our approach is to \emph{estimate} even \emph{local degrees} of vertices (in contrast to computing them exactly). To obtain the estimates, it suffices to examine only small neighborhoods of vertices and in turn preserve the independent distribution of the intact ones. More precisely, we sample a small set of vertices on each machine, called \emph{reference sets}, and use the set to estimate the local degrees of all vertices assigned to this machine. Furthermore, we show that with a proper adjustments of \globalAlgorithm{} these estimates are sufficient for capturing high-degree vertices.

Very crucially, all the vertices that are used in computing a matching in one emulated phase (including the reference sets) are discarded at the end of the phase, even if they do not participate in the obtained matching. In this way we disregard the vertices which position is fixed and, intuitively, secure an independent distribution of the vertices across the machines in the next phase.

 We also note, without going into details, that obtaining full independence required modifying how the set of friends is selected, compared to the original approach of Onak and Rubinfeld \cite{OR10}. In their approach, each heavy vertex selected one friend at random. However, as before, in order to select exactly one friend would require examining neighborhood of heavy vertices. This, however, introduces dependencies between vertices that have not been selected. So instead, in our \globalAlgorithm{}, every vertex selects itself as a friend independently and proportionally to the number of high-degree vertices (found using the reference set), which again secures an independent distribution of the remaining vertices. The final properties of the obtained sets in either approach are very similar.

 \subsubsection{Uniformity (Lemma~\ref{lem:preserving_near_uniformity})}
A very convenient property in the task of emulating multiple phases of $\globalAlgorithm$ is a uniform distribution of vertices across all the machines at every phase -- for such a distribution, we know the expected number of neighbors of each desired type assigned to the same machine. Obtaining perfect uniformity seems difficult---if not impossible in our setting---and we therefore settle for \emph{near} uniformity of vertex assignments. The probability of the assignment of each vertex to each machine is allowed to differ slightly from that in the uniform distribution. Initially, the distribution of each vertex is uniform and with every phase it can deviate more and more from the uniform distribution. We bound the rate of the decay with high probability and execute multiple rounds as long as the deviation from the uniform distribution is negligible.
 More precisely, in the execution of the entire parallel algorithm, the sufficiently uniform distribution is on average kept over $\Omega\left(\tfrac{\log{n}}{\left(\log \log n\right)^2}\right)$ phases of the emulation of $\globalAlgorithm$.

\begin{figure}[ht]
\begin{center}
\includegraphics{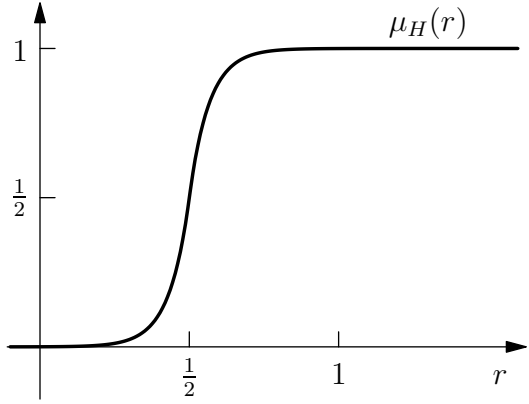}
\end{center}
\caption{An idealized version of $\prH:\mathbb R\to [0,1]$, in which $n$ was fixed to a small constant and the multiplicative constant inside the exponentiation operator was lowered.}
\label{fig:prH}
\end{figure}
 In order to achieve the near uniformity, we modify the procedure for selecting $H$, the set of high-degree vertices. Instead of a hard threshold on the degrees of vertices that are included in $H$ as in the sequential algorithm, we randomize the selection by using a carefully crafted threshold function $\prH$. This function specifies the probability with which a vertex is included in $H$. It takes as input the ratio of the vertex's degree to the current maximum degree (or, more precisely, the current upper bound on the maximum degree) and it smoothly transitions from 0 to 1 in the neighborhood of the original hard threshold (see Figure~\ref{fig:prH}). The main intuition behind the introduction of this function is that we want to ensure that a vertex is \emph{not} selected for $H$ with almost the same probability, independently of the machine on which it resides. Using a hard threshold instead of $\prH$ could result in the following deficiency. Consider a vertex $v$ that has slightly too few neighbors to qualify as a heavy vertex. Still, it could happen, with a non-negligible probability, that the reference set of some machine contains so many neighbors of $v$ that $v$ would be considered heavy on this machine. However, if $v$ is not included in the set of heavy vertices on that machine, it becomes clear after even a single phase that the vertex is not on the given machine, i.e. the vertex is on the given machine with probability zero. At this point the distribution is clearly no longer uniform.

 Function $\prH$ has further useful properties that we extensively exploit in our analysis. We just note that in order to ensure near uniformity with high probability, we also have to ensure that each vertex is selected for $F$, the set of friends, with roughly the same probability on each machine.

\subsection{Organization of the appendix}
The appendix presents the details of our techniques.
We start by analyzing $\globalAlgorithm$ in Section~\ref{sec:global-algorithm}.
Then, Section~\ref{sec:emulate-phase} describes how to emulate of a single phase of $\globalAlgorithm$ in the MPC model.
Section~\ref{sec:parallel} gives and analyzes our parallel algorithm by putting together components developed in the previous sections.
The resulting parallel algorithm can be implemented in the MPC model in a fairly straightforward way by using the result of~\cite{goodrich2011sorting}.
The details of the implementation are given in Section~\ref{sec:complexity-and-implementation}.

\section{Global Algorithm}
\label{sec:global-algorithm}
\newcommand{\expe}[1]{\Exp\left[#1\right]}
\newcommand{\card}[1]{\left|#1\right|}

\subsection{Overview}
The starting point of our result is a peeling algorithm \globalAlgorithm{} that takes as input a graph $G$, and removes from it vertices of lower and lower degree until no edge is left. See page~\pageref{alg:global} for its pseudocode. We use the term \emph{phase} to refer to an iteration of the main loop in Lines~\ref{line:global:main_loop_start}--\ref{line:global:main_loop_end}.

Each phase is associated with a threshold $\threshold$. Initially, $\threshold$ equals $\widetilde \threshold$, the upper bound on the maximum vertex degree. In every phase,  $\threshold$ is divided by two until it becomes less than one and the algorithm stops. Since during the execution of the algorithm we maintain the invariant that the maximum degree in the graph is at most $\threshold$, the graph has no edge left when the algorithm terminates.

In each phase the algorithm matches, in expectation, a constant fraction of the vertices it removes. We use this fact to prove that, across all the phases, the algorithm computes a constant-factor approximate matching.

We now describe in more detail the execution of each phase. First, the algorithm creates $H$, the set of vertices that have degree at least $\threshold/2$ (Line~\ref{line:global:high-degree}). We call these vertices \emph{heavy}. Then, the algorithm uses randomness to create $F$, a set of \emph{friends} (Line~\ref{line:global:friends}). Each vertex $v$ is independently included in $F$ with probability equal to the number of its heavy neighbors divided by $4\threshold$. We show that $\expe{\card{F}} = O(\card{H})$ and $G[H \cup F]$ contains a matching of expected size $\Omega(\card{H})$. This kind of matching is likely found by \matchHeavy{} in Line~\ref{line:compute_matching}.

Note that \globalAlgorithm{} could as well compute a maximal matching in $G[H\cup F]$ instead of calling \matchHeavy{}. However, for the purpose of the analysis, using \matchHeavy{} is simpler, as we can directly relate the size of the obtained matching to the size of $H$. In addition, we later give a parallel version of \globalAlgorithm{}, and \matchHeavy{} is easy to parallelize. 

At the end of the phase, vertices in both $H$ and $F$ are removed from the graph, while the matching found in $G[H\cup F]$ is added to the global matching being constructed. It is easy to see, that by removing $H$, the algorithm ensures that no vertex of degree larger than $\threshold/2$ remains in the graph, and therefore the bound on the maximum degree decreases by a factor of two.

\subsection{Analysis}

We start our analysis of the algorithm by showing that the execution of \matchHeavy{} in each phase of \globalAlgorithm{} finds a relatively large matching in expectation.

\begin{lemma}\label{lem:global_single_round}
Consider one phase of \globalAlgorithm.
Let $H$ be the set of heavy vertices. \matchHeavy{} finds a matching $\widetilde M$ such that $\Exp\left[\left|\widetilde M\right|\right] \ge \frac{1}{40} |H|$.
\end{lemma}

\begin{proof}
Observe that the set $E_\star$ is a collection of vertex-disjoint stars: each edge connects a red vertex with a blue vertex and the red vertices have degree $1$. Thus, a subset of $E_\star$ forms a valid matching as long as no blue vertex is incident to two matched edges. Note that this is guaranteed by how edges are added to $\widetilde M$ in Line~\ref{step:match}.

The size of the computed matching is the number of blue vertices in $H$ that have at least one incident edge in $E_\star$. Let us now lower bound the number of such vertices. Consider an arbitrary $u \in H$. It has the desired properties exactly when the following three independent events happen: some $v \in F$ selects $u$ in Line~\ref{step:select}, $u$ is colored blue, and $v$ is colored red.
The joint probability of the two latter events is exactly $\frac14$.
The probability that $u$ is \emph{not} selected by some $v \in F$ is
\[\left(1-\frac{1}{4\threshold}\right)^{|N(u) \cap V'|}
\le \left(1-\frac{1}{4\threshold}\right)^{\threshold/2}
\le \exp\left(-\frac{1}{4\threshold} \cdot \frac{\threshold}{2}\right) \le \exp\left(-\frac{1}{8}\right) \le \frac{9}{10}.\]
This implies that $u$ is selected by a neighbor $v \in F$ with probability at least $\frac{1}{10}$. Therefore, with probability at least $\frac{1}{10} \cdot \frac{1}{4} = \frac{1}{40}$, $u$ is blue and incident to an edge in $E_\star$. Hence, $\expe{\card{\widetilde M}} \geq \frac{1}{40}\card{H}$.
\end{proof}

Next we show an upper bound on the expected size of $F$, the set of friends.

\begin{lemma}\label{lem:hf}
Let $H$ be the set of heavy vertices selected in a phase of \globalAlgorithm{}. The following bound holds on the expected size of $F$, the set of friends, created in the same phase: $\Exp\left[\left|F\right|\right] \le \frac{1}{4}|H|$.
\end{lemma}

\begin{proof}
At the beginning of a phase,
every vertex $u \in V'$---including those in $H$---has its degree, $|N(u) \cap V'|$, bounded by $\threshold$. Reversing the order of the summation and applying this fact, we get:
\[\Exp\left[|F|\right] = \sum_{v \in V'} \frac{|N(v) \cap H|}{4\threshold} = \sum_{u \in H} \frac{\left|N(u) \cap V'\right|}{4\threshold}
\le \frac{|H|\cdot \threshold}{4\threshold} = \frac{|H|}{4}.\qedhere\]
\end{proof}

We combine the last two bounds to lower bound the expected size of the matching computed by \globalAlgorithm.

\begin{lemma}
\label{lem:global_alg_expectation}
        Consider an input graph $G$ with an upper bound $\widetilde\threshold$ on the maximum vertex degree.
	$\globalAlgorithm(G,\widetilde\threshold)$ executes $T \eqdef \lfloor\log \widetilde\threshold\rfloor+1$ phases. Let $H_i$, $F_i$, and $\widetilde M_i$ be the sets $H$, $F$, and $\widetilde M$ constructed in phase $i$ for $i \in [T]$.
	The following relationship holds on the expected sizes of these sets:
	\[\sum_{i=1}^{T}\expe{\card{\widetilde M_i}} \ge \frac{1}{50} \sum_{i=1}^{T}\expe{\card{H_i} + \card{F_i}} \]
\end{lemma}

\begin{proof}
For each phase $i \in [T]$, by applying the expectation over all possible settings of the set $H_i$, we learn from Lemmas~\ref{lem:global_single_round} and~\ref{lem:hf} that 
\[
\expe{\card{\widetilde M_i}} \ge \frac{1}{40} \expe{\card{H_i}}
\qquad
\mbox{and}
\qquad
\expe{\card{F_i}} \le \frac{1}{4} \expe{\card{H_i}}.
\]
It follows that
\[\frac{1}{50}\expe{\card{H_i} + \card{F_i}}
\le 
\frac{1}{50}\expe{\card{H_i}} + \frac{1}{200}\expe{\card{H_i}}
=
\frac{1}{40}\expe{\card{H_i}}
\le \expe{\card{\widetilde M_i}},
\]
and the statement of the lemma follows by summing over all phases.
\end{proof}

We do not use this fact directly in our paper, but note that the last lemma can be used to show that \globalAlgorithm{} can be used to find a large matching.

\begin{corollary}\label{cor:global_alg_is_good}
\globalAlgorithm{} computes a constant factor approximation to the maximum matching with $\Omega(1)$ probability.
\end{corollary}

\begin{proof}
First, note that \globalAlgorithm{} finds a correct matching, i.e., no two different edges in $M$ share an endpoint. This is implied by the fact that $M$ is extended in every phase by a matching on a disjoint set of vertices.

Let $T$ and sets $H_i$, $F_i$, and $\widetilde M_i$ for $i \in [T]$ be defined as in the statement of Lemma~\ref{lem:global_alg_expectation}.
Let $M_{\rm OPT}$ be a maximum matching in the graph. Observe that at the end of the algorithm execution, the remaining graph is empty. This implies that the size of the maximum matching can be bounded by the total number of removed vertices, because each removed vertex decreases the maximum matching size by at most one:
\[
\sum_{i=1}^T \card{H_i} + \card{F_i} \ge \card{M_{\rm OPT}}.
\]
Hence, using Lemma~\ref{lem:global_alg_expectation},
\[
\expe{\card{M}}
= 
\sum_{i=1}^{T}\expe{\card{\widetilde M_i}}
\ge \frac{1}{50} \sum_{i=1}^{T}\expe{\card{H_i} + \card{F_i}}
\ge \frac{1}{50}\card{M_{\rm OPT}}.
\]
Since $\card{M} \le \card{M_{\rm OPT}}$, $|M| \ge \frac{1}{100} \card{M_{\rm OPT}}$ with probability at least $\frac{1}{100}$. Otherwise, $\expe{\card{M}}$ would be strictly less than $\frac{1}{100} \cdot \card{M_{\rm OPT}} + 1 \cdot \frac{1}{100}\card{M_{\rm OPT}} = \frac{1}{50}\card{M_{\rm OPT}}$, which is not possible.
\end{proof}

\section{Emulation of a Phase in a Randomly Partitioned Graph}
\label{sec:emulate-phase}

In this section, we introduce a modified version of a single phase (one iteration of the main loop) of \globalAlgorithm{}. Our modifications later allow for implementing the algorithm in the MPC model. The pseudocode of the new procedure, \globalPhase, is presented as Algorithm~\ref{alg:full_phase}. We partition the vertices of the current graph into $\machines$ sets $V_i$, $1\le i \le \machines$. Each vertex is assigned independently and \emph{almost} uniformly at random to one of the sets. For each set $V_i$, we run a subroutine \localPhase{} (presented as Algorithm~\ref{alg:local_phase}). This subroutine runs a carefully crafted approximate version of one phase of \globalAlgorithm{} with an appropriately rescaled threshold $\threshold$. More precisely, the threshold passed to the subroutine is scaled down by a factor of $\machines$, which corresponds to how approximately vertex degrees decrease in subgraphs induced by each of the sets. The main intuition behind this modification is that we hope to break the problem up into smaller subproblems on disjoint induced subgraph, and obtain similar \emph{global} properties by solving the problem approximately on each smaller part. Later, in Section~\ref{sec:parallel}, we design an algorithm that assigns the subproblems to different machines and solves them in parallel.

\newsavebox{\bbooxx}
\savebox{\bbooxx}{\phantom{\bf Algorithm 5:}}

\begin{algorithm}[ht]\algfontsize
  \caption{$\globalPhase(\threshold,G_\star,\machines,\DD)$
  \algdesc{Emulation of a single phase in a randomly partitioned graph}}\label{alg:full_phase}
  \KwIn{\\
      \qquad$\bullet$ threshold $\threshold$\\
      \qquad$\bullet$ induced subgraph $G_\star = (V_\star,E_\star)$ of maximum degree $\frac{3}{2}\threshold$\\
      \qquad$\bullet$ number $\machines$ of subgraphs\\
      \qquad$\bullet$ $\eps$-near uniform and independent distribution $\DD$ on assignments of $V_\star$ to $[\machines]$\\
      }
  \KwOut{Remaining vertices and a matching}
  \BlankLine

  Pick a random assignment $\Phi:V_\star \to [\machines]$ from $\DD$\label{line:full_phase:pick_assignment}

  \For{$i \in [\machines]$}{

       $V_i \leftarrow \{v \in V_\star : \Phi(v) = i\}$

       $(V'_i,M_i) \leftarrow \localPhase(i, G_\star[V_i], \threshold / \machines)$ \qquad\tcc{\localPhase{} = Algorithm~\ref{alg:local_phase}}

  }

  \Return{$\left(\bigcup_{i=1}^{\machines} V'_i, \bigcup_{i=1}^{\machines} M_i \right)$}
	\label{line:full-phase-return}
\end{algorithm}

\begin{table}[ht]\algfontsize
\begin{mdframed}[style=myframe]
A multiplicative constant used in the exponent of $\prH$:
$$\prHmult \eqdef 96 \ln n.$$

The probability of the selection for a reference set:
  \[\prR \eqdef \left(10^6\cdot\log n\right)^{-1}.\]

The probability of the selection for a heavy set (used with $r$ equal to the ratio of the estimated degree to the current threshold):
  \[\prH(r) \eqdef
  \begin{cases}
  \frac{1}{2}\exp\left(\frac{\prHmult}{2} \left(r-1/2\right)\right)&
  \text{if $r \le 1/2$,}\\
  1-\frac{1}{2}\exp\left(-\frac{\prHmult}{2} \left(r-1/2\right)\right)&
  \text{if $r > 1/2$.}
  \end{cases}\]

The probability of the selection for the set of friends (used with $r$ equal to the ratio of the number of heavy neighbors to the current threshold):
  \[\prF(r) \eqdef
  \begin{cases}
  \max\{r/4,0\}&
  \text{if $r \le 4$,}\\
  1&
  \text{if $r > 4$.}
  \end{cases}\]

\end{mdframed}
\caption{Global parameters $\prHmult \in (1,\infty)$ and $\prR \in (0,1)$ and functions $\prH:\R\to[0,1]$ and $\prF:\R\to[0,1]$ used in the parallel algorithm. $\prHmult$, $\prR$, and $\prH$ depend on $n$, the total number of vertices in the graph.}
\label{table:parameters}
\end{table}

We now discuss \localPhase{} (i.e., Algorithm~\ref{alg:local_phase}) in more detail. Table~\ref{table:parameters} introduces two parameters, $\prHmult$ and $\prR$, and two functions, $\prH$ and $\prF$, which are used in \localPhase. Note first that $\prHmult$ is a parameter used in the definition of $\prH$ but it is not used in the pseudocode of \localPhase{} (or \globalPhase) for anything else. It is, however, a convenient abbreviation in the analysis and the later parallel algorithm. The other three mathematical objects specify probabilities with which vertices are included in sets that are created in an execution of \localPhase.

\begin{algorithm}[ht]\algfontsize
  \caption{$\localPhase(i,G_i,\threshold_\star)$\algdesc{Emulation of a single phase on an induced subgraph}}
  \label{alg:local_phase}
  \KwIn{\\
      \qquad$\bullet$ induced subgraph number $i$ (useful only for the analysis)\\
      \qquad$\bullet$ induced subgraph $G_i = (V_i,E_i)$\\
      \qquad$\bullet$ threshold $\threshold_\star \in \mathbb R_+$
      }
  \KwOut{Remaining vertices and a matching on $V_i$}
  \BlankLine
     Create a \emph{reference set} $R_i$ by independently selecting each vertex in $V_i$ with probability $\prR$.
     \label{line:local_phase:create_R}

     For each $v \in V_i$, $\degEst{v} \leftarrow |N(v) \cap R_i| / \prR$.
     \label{line:local_phase:estimate_degrees}

     Create a set $H_{i}$ of \emph{heavy vertices} by independently selecting each $v \in V_i$ with probability $\prH\left(\degEst{v} / \threshold_\star\right)$.
     \label{line:local_phase:create_H}

     Create a set $F_{i}$ of \emph{friends} by independently selecting each vertex in $v \in V_i$ with probability $\prF\left(|N(v) \cap H_{i}| / \threshold_\star\right)$.
     \label{line:local_phase:create_F}

     Compute a maximal matching $M_i$ in $G[H_i \cup F_i]$.
     \label{line:local_phase:matching}

   \Return{$\left(V_i \setminus (R_i \cup H_i \cup F_i),M_i\right)$}
	\label{line:local-phase-return}
\end{algorithm}
Apart from creating its own versions of $H$, the set of heavy vertices, and $F$, the set of friends, \localPhase{} constructs also a set $R_i$, which we refer to as a \emph{reference set}. In Line~\ref{line:local_phase:create_R}, the algorithm puts each vertex in $R_i$ independently and with the same probability $\prR$. The reference set is used to estimate the degrees of other vertices in the same induced subgraph in Line~\ref{line:local_phase:estimate_degrees}. For each vertex $v_i$, its estimate $\degEst{v}$ is defined as the number of $v$'s neighbors in $R_i$ multiplied by $\prR^{-1}$ to compensate for sampling. Next, in Line~\ref{line:local_phase:create_H}, the algorithm uses the estimates to create $H_i$, the set of heavy vertices. Recall that $\globalAlgorithm$ uses a sharp threshold for selecting heavy vertices: all vertices of degree at least $\threshold/2$ are placed in $H_i$. \localPhase{} works differently. It divides the degree estimate by the current threshold $\threshold_\star$ and uses function $\prH$ to decide with what probability the corresponding vertex is included in $H_i$. A sketch of the function can be seen in Figure~\ref{fig:prH}. The function transitions from almost $0$ to almost $1$ in the neighborhood of $\frac{1}{2}$ at a limited pace. As a result vertices of degrees smaller than, say, $\frac{1}{4}\threshold$ are very unlikely to be included in $H_i$ and vertices of degree greater than $\frac{3}{4}\threshold$ are very likely to be included in $H_i$.
\globalAlgorithm{} can be seen as an algorithm that instead of $\prH$, uses a step function that equals $0$ for arguments less than $\frac{1}{2}$ and abruptly jumps to $1$ for larger arguments. Observe that without $\prH$, the vertices whose degrees barely qualify them as heavy could behave very differently depending on which set they were assigned to. We use $\prH$ to guarantee a smooth behavior in such cases. That is one of the key ingredients that we need for making sure that a set of vertices that remains on one machine after a phase has almost the same statistical properties as a set of vertices obtained by new random partitioning.

Finally, in Line~\ref{line:local_phase:create_F}, \localPhase{} creates a set of friends. This step is almost identical to what happens in the global algorithm. The only difference is that this time we have no upper bound on the number of heavy neighbors of a vertex. As a result that number divided by $4\threshold_\star$ can be greater than 1, in which case we have to replace it with 1 in order to obtain a proper probability. This is taken care of by  function $\prF$. Once $H_i$ and $F_i$ have been created, the algorithm finds a maximal matching $M_i$ in the subgraph induced by the union of these two sets. The algorithm discards from the further consideration not only $H_i$ and $F_i$, but also $R_i$. This eliminates dependencies in the possible distribution of assignments of vertices that have not been removed yet if we condition this distribution on the configuration of sets that have been removed. Intuitively, the probability of a vertex's inclusion in any of these sets depends only on $R_i$ and $H_i$ but not on any other
vertices. Hence, once we fix the sets of removed vertices, the assignment of the remaining vertices to subgraphs is fully independent.\footnote{By way of comparison, consider observing an experiment in which we toss the same coin twice. The bias of the coin is not fixed but comes from a random distribution. If we do not know the bias, the outcomes of the coin tosses are not independent. However, if we do know the bias, the outcomes are independent, even though they have the same bias.} The output of \localPhase{} is a subset of $V_i$ to be considered in later phases and a matching $M_i$, which is used to expand the matching that we construct for the entire input graph.
\newcommand{\family}[1]{\{#1_{i}\}_{i \in [\machines]}}%
We now introduce additional concepts and notation. They are useful for describing and analyzing properties of the algorithm. A configuration describes sets $R_i$, $H_i$, and $F_i$, for $1 \le i \le \machines$, constructed in an execution of \globalPhase. We use it for conditioning a distribution of vertex assignments as described in the previous paragraph. We also formally define two important properties of distributions of vertex assignments: independence and near uniformity.

\paragraph{Configurations.}
Let $\machines$ and $V_\star$ be the parameters to \globalPhase: the number of subgraphs and the set of vertices in the graph to be partitioned, respectively.
We say that
\[\CC = \left(\family{R},\family{H},\family{F}\right)\]
is an \emph{$\machines$-configuration} if it represents a configuration of sets $R_{i}$, $H_{i}$, and $F_{i}$ created by \globalPhase{} in the simulation of a phase. Recall that for any $i \in [\machines]$, $R_{i}$, $H_{i}$, and $F_{i}$ are the sets created (and removed) by the execution of \localPhase{} for $V_i$, the $i$-th subset of vertices.

We say that a vertex $v$ is \emph{fixed} by $\CC$ if it belongs to one of the sets in the configuration, i.e.,
\[v \in \bigcup_{i \in [\machines]}\left(R_{i} \cup H_{i} \cup F_{i}\right).\]

\paragraph{Conditional distribution.}
Let $\DD$ be a distribution on assignments $\varphi:V_\star \to [\machines]$. Suppose that we execute \globalPhase{} for $\DD$
and let $\CC$ be a \emph{non-zero probability} $\machines$-configuration---composed of sets $R_{i}$, $H_{i}$, and $F_{i}$ for $i \in [\machines]$---that can
be created in this setting. Let $V'_\star$ be the set of vertices in $V_\star$ that are \emph{not} fixed by $\CC$.
We write $\DD[\CC]$ to denote the conditional distribution of possible assignments of vertices in $V'_\star$ to $[\machines]$, given that for all $i\in[\machines]$, $R_{i}$, $H_{i}$, and $F_{i}$ in $\CC$ were the sets constructed by \localPhase{} for the $i$-th induced subgraph.

\paragraph{Near uniformity and independence.}
Let $\DD$ be a distribution on assignments $\varphi:\widetilde V \to [\machines]$ for some set $\widetilde V$ and $\machines$.
For each vertex $v \in \widetilde V$, let $p_v:[\machines] \to [0,1]$ be the probability mass function of the marginal distribution of $v$'s assignment. For any $\eps \ge 0$, we say that $\DD$ is \emph{$\eps$-near uniform} if for every vertex $v$ and every $i \in [\machines]$, $p_v(i) \in \range{(1\pm\eps)/\machines}$. We say that $\DD$ is an \emph{independent} distribution if the probability of every assignment $\varphi$ in $\DD$
equals exactly $\prod_{v \in V'} p_v(\varphi(v))$.

\paragraph{Concentration inequality.}
We use the following version of the Chernoff bound that depends on an upper bound on the expectation of the underlying independent random variables.
It can be shown by combining two applications of the more standard version.

\begin{lemma}[Chernoff bound]
\label{lem:chernoff_consequence}
\newcommand{\bound}{U}
Let $X_1$, \ldots, $X_k$ be independently distributed random variables taking values in $[0,1]$. Let $X \eqdef X_1 + \cdots + X_k$ and let $\bound \ge 0$ be an upper bound on the expectation of $X$, i.e., $\Exp[X] \le \bound$. For any $\delta \in [0,1]$, $\Pr(|X - \Exp[X]| > \delta \bound) \le 2\exp(-\delta^2 \bound/3)$.
\end{lemma}

\paragraph{Concise range notation.} Multiple times throughout a paper, we want to denote a range around some value. Instead of writing, say, $[x-\delta,x+\delta]$, we introduce a more concise notation. In this specific case, we would simply write $\range{x\pm \delta}$. More formally, let $E$ be a numerical expression that apart from standard operations also contains a single application of the binary or unary operator $\pm$. We create two standard numerical expressions from $E$: $E_{-}$ and $E_{+}$ that replace ${\pm}$ with ${-}$ and ${+}$, respectively. Now we define $\range{E} \eqdef [\min\{E_-,E_+\},\max\{E_-,E_+\}]$.

As another example, consider $E = \sqrt{101\pm 20}$. We have $E_{-} = \sqrt{101-20}= 9$ and $E_{+} = \sqrt{101+20} = 11$. Hence $\range{\sqrt{101\pm 20}} = [\min\{9,11\},\max\{9,11\}] = [9,11]$.

We now show the properties of \globalPhase{} that we use to obtain our final parallel algorithm.

\subsection{Outline of the section}
We start by showing that \globalPhase{} computes a large matching as follows. Each vertex belonging to $H_i$ or $F_i$ that \globalPhase{} removes in the calls to \localPhase{} can decrease the maximum matching size in the graph induced by the remaining vertices by one. We show that the matching that \globalPhase{} constructs in the process captures on average at least a constant fraction of that loss. We also show that the effect of removing $R_i$ is negligible. More precisely, in Section~\ref{sec:expected-matching-size} we prove the following lemma.
\begin{restatable}{lemma}{lemmamatchingsize}\label{lem:modified_phase_matching_size}
Let $\threshold$, $G_\star = (V_\star,E_\star)$, $\machines$, and $\DD$ be parameters for \globalPhase{} such that
\begin{itemize}

	\item $\DD$ is an independent and $\eps$-near uniform distribution on assignments of vertices $V_\star$ to $[\machines]$
for $\eps \in [0,1/200]$,

\item $\frac{\threshold}{\machines} \ge 4000\prR^{-2}\ln^2n $,

\item the maximum degree of a vertex in $G_\star$ is at most $\frac{3}{2}\threshold$.
\end{itemize}
For each $i \in [\machines]$, let $H_i$, $F_i$, and $M_i$ be the sets constructed by \localPhase{} for the $i$-th induced subgraph.
Then, the following relationship holds for their expected sizes:
\[\sum_{i \in [\machines]}\Exp\left[\left|H_{i}\cup F_{i}\right|\right]
\le n^{-9} + 1200 \sum_{i \in [\machines]}\Exp\left[\left|M_{i}\right|\right]. \]
\end{restatable}

Note that Lemma~\ref{lem:modified_phase_matching_size} requires that the vertices are distributed independently and near uniformly in the $m$ sets. This is trivially the case right after the vertices are partitioned independently at random.
However, in the final algorithm, after we partition the vertices, we run \emph{multiple} phases on each machine. In the rest of this section we show that running a single phase \emph{preserves} independence of vertex distribution and only slightly disturbs the uniformity (Lemma~\ref{lem:independence} and Lemma~\ref{lem:preserving_near_uniformity}). As we have mentioned before, independence stems from the fact that we use reference sets to estimate vertex degrees. We discard them at the end and condition on them, which leads to the independence of the distribution of vertices that are not removed.
\begin{restatable}{lemma}{lemmaindependence}\label{lem:independence}
Let $\DD$ be an independent distribution of assignments of vertices in $V_\star$ to $[\machines]$. Let $\CC$ be a non-zero probability $\machines$-configuration that can be constructed by \globalPhase{} for $\DD$. Let $V'_\star$ be the set of vertices of $V_\star$ that are not fixed by $\CC$. Then $\DD[\CC]$ is an independent distribution of vertices in $V'_\star$ on $[\machines]$.
\end{restatable}

Independence of the vertex assignment is a very handy feature that allows us to use Chernoff-like concentration inequalities in the analysis of multiple phase emulation. However, although the vertex assignment of non-removed vertices remains independent across machines from phase to phase, as stated by Lemma~\ref{lem:independence}, their distribution is not necessarily uniform. Fortunately, we can show it is near uniform.
\\
The proof of near uniformity is the most involved proof in this paper. In a nutshell, the proof is structured as follows. We pick an arbitrary vertex $v$ that has not been removed and show that with high probability it has the same number of neighbors in all sets $R_i$. The same property holds for $v$'s neighbors in all sets $H_i$. We use this to show that the probability of a fixed configuration of sets removed in a single phase is roughly the same for all assignments of $v$ to subgraphs. In other words, if $v$ was distributed nearly uniformly before the execution of \globalPhase, it is distributed only slightly less uniformly after the execution.
\begin{restatable}{lemma}{lemmauniformity}\label{lem:preserving_near_uniformity}
Let $\threshold$, $G_\star = (V_\star,E_\star)$, $\machines$, and $\DD$ be parameters for \globalPhase{} such that
\begin{itemize}

\item $\DD$ is an independent and $\eps$-near uniform distribution on assignments of vertices $V_\star$ to $[\machines]$
for $\eps \in [0,(200 \ln n)^{-1}]$,

\item $\frac{\threshold}{\machines} \ge 4000\prR^{-2}\ln^2n $.

\end{itemize}
Let $\CC$ be an $m$-configuration constructed by \globalPhase. With probability at least $1-n^{-4}$ both the following properties hold:
\begin{itemize}
 \item The maximum degree in the graph induced by the vertices not fixed in $\CC$ is bounded by $\frac{3}{4}\threshold$.
 \item $\DD[\CC]$ is $60\prHmult \left(\left(\frac{\threshold}{\machines}\right)^{-1/4}+\eps\right)$-near uniform.
\end{itemize}

\end{restatable}

\subsection{Expected matching size}
\label{sec:expected-matching-size}
Now we prove Lemma~\ref{lem:modified_phase_matching_size}, i.e. we show that \globalPhase{} computes a large matching.
In the proof we argue that the expected total size of sets $H_i$ and $F_i$ is not significantly impacted by relatively low-degree vertices classified as heavy or by an unlucky assignment of vertices to subgraphs resulting in local vertex degrees not corresponding to global degrees.
Namely, we show that the expected number of friends a heavy vertex adds is $O(1)$ and at the same time the probability that the vertex gets matched is $\Omega(1)$.

\lemmamatchingsize*
\newcommand{\curG}{\textcolor{cyan}{G_\diamond}}%
\newcommand{\curV}{\textcolor{cyan}{V_\diamond}}%
\begin{proof}
\newcommand{\tM}{\widetilde{M}}
We borrow more notation from \globalPhase{} and the $m$ executions of \localPhase{} initiated by it.
For $i \in [\machines]$, $V_i$ is the set inducing the $i$-th subgraph. Value $\threshold_\star = \frac{\threshold}{\machines}$ is the rescaled threshold passed to the executions of \localPhase{}. $R_i$ is the reference set created by \localPhase{} for the $i$-th induced subgraph.

For each induced subgraph, \localPhase{} computes a maximal matching $M_i$ in Line~\ref{line:local_phase:matching}. While such a matching is always large---its size is at least half the maximum matching size---it is hard to relate its size directly to the sizes of $H_i$ and $F_i$. Therefore, we first analyze the size of a matching that would be created by $\matchHeavy(G_\star[H_i \cup F_i], H_i, F_i)$. We refer to this matching as $\tM_i$ and we later use the inequality $\left|\tM_i\right| \le 2\left|M_i\right|$.

We partition each $H_i$, $i \in [\machines]$, into two sets: $H'_{i}$ and $H''_{i}$.
$H'_{i}$ is the subset of vertices in $H_{i}$ of degree less than $\frac{1}{8}\threshold$ in $G_\star$. $H''_{i,t+1}$ is its complement, i.e., $H''_{i} \eqdef H_{i} \setminus H'_{i}$.
We start by bounding the expected total size of sets $H'_{i}$. What is the probability that a given vertex $v$ of degree less than $\frac{1}{8}\threshold$ is included in $\bigcup_{i \in [\machines]} H_{i}$?
	Suppose that $v \in V_k$, where $k \in [\machines]$.
The expected number of $v$'s neighbors in $R_k$ is at most
$(1+\eps) \cdot \prR \cdot \frac{1}{8} \threshold/\machines \le \frac{3}{16}\prR \threshold_\star$
due to the independence and $\eps$-near uniformity of $\DD[\CC]$.
Using the independence, Lemma~\ref{lem:chernoff_consequence}, and the lower bound on $\threshold_\star$, we obtain the following bound:
\[
\Pr\left[\prR\degEst{v} > \frac{1}{4}\prR\threshold_\star\right]
\le 2 \exp\left(-\frac{1}{3} \cdot \left(\frac{1}{3}\right)^2  \cdot \frac{3}{16} \prR \threshold_\star\right)
\le 2 \exp\left(-27 \ln n\right) = 2n^{-27}.
\]
If $\degEst{v} \le \frac{1}{4}\threshold_\star$, the probability that $v$ is selected to $H_k$ is at most $\prH(\degEst{v}/\threshold_\star) \le \prH(1/4) \le \frac{1}{2}n^{-12}$. Hence $v$ is selected to $H_k$---and therefore to $H'_k$---with probability at most $
2n^{-27} + \frac{1}{2}n^{-12} \le
n^{-12}$. This implies that $ \sum_{i \in [\machines]} \Exp\left[\left| H'_{i} \right| \right] \le n \cdot n^{-12} = n^{-11}$.

We also partition the sets of friends, $F_{i}$ for $i \in [\machines]$, into two sets each: $F'_i$ and $F''_i$. This partition is based on the execution of \matchHeavy{} for the $i$-th subgraph. In Line~\ref{line:randomized_matching:v_star}, this algorithm selects for every vertex $v \in F_i$ a random heavy neighbor $v_\star \in H_i$. If $v_\star \in H'_i$, we assign $v$ to $F'_i$. Analogously, if $v_\star \in H''_i$, we assign $v$ to $F''_i$.
Obviously, a heavy vertex in $H'_{i}$ can be selected only if $H'_{i}$ is non-empty. By Markov's inequality and the upper bound on $ \sum_{i \in [\machines]} \Exp\left[\left| H'_{i} \right| \right]$, the probability that at least one set $H'_{i}$ is non-empty is at most $n^{-11}$. Even if for all $i\in[\machines]$, all vertices in $F_i$ select a heavy neighbor in $H'_{i}$ whenever it is available,
the total expected number of vertices in sets $F'_{i}$ is at most $ \sum_{i \in [\machines]} \Exp\left[\left| F'_{i,t+1} \right| \right] \le n \cdot n^{-11} = n^{-10}$.

Before we proceed to bounding sizes of the remaining sets, we prove that with high probability, all vertices have a number of neighbors close to the expectation.
Let $\varphi:V_\star \to [\machines]$ be the assignment of vertices to subgraphs. We define $\EE$ as the event that for all $v \in V_\star$,
\[
\left|\frac{1}{m}\left|N(v) \cap V_\star\right| - \left|N(v) \cap V_{\varphi(v)}\right| \right|
\le
\frac{1}{16}\threshold_\star.
\]
Consider first one fixed $v \in V_\star$. The degree of $v$ in $G_\star$ is
$\left|N(v) \cap V_\star\right| \le \frac{3}{2}\threshold$.
Due to the near-uniformity and independence,
\[
\left|\frac{1}{m}\left|N(v) \cap V_\star\right| - \Exp\left[\left|N(v) \cap V_{\varphi(v)}\right|\right] \right|
\le
\eps \cdot \frac{3}{2} \frac{\threshold}{\machines} \le \frac{3}{400}\threshold_\star.
\]
This in particular implies that $\Exp\left[\left|N(v) \cap V_{\varphi(v)}\right|\right] \le \left(\frac{3}{2} + \frac{3}{400}\right)\threshold_\star \le 2\threshold_\star$.
Using the independence of $\DD$, Lemma~\ref{lem:chernoff_consequence}, and the lower bound on $\threshold_\star$ (i.e., $\threshold_\star = \tfrac{\threshold}{\machines} \ge 4000\prR^{-2}\ln^2n = 4\cdot 10^{15} \cdot \ln^4 n$),
\begin{align*}
\Pr\left[ \left| \Exp\left[\left|N(v) \cap V_{\varphi(v)}\right|\right] - \left|N(v) \cap V_{\varphi(v)}\right| \right| > \frac{1}{20}\threshold_\star \right]
& \le 2\exp\left(-\frac{1}{3} \cdot \left(\frac{1}{20} \cdot \frac{1}{2}\right)^2 \cdot 2\threshold_\star\right)\\
& \le 2\exp\left(-(10^{12} + 3) \ln n\right)\\
& \le n^{-(10^{12} + 2)} \le n^{-12}.
\end{align*}
As a result, with this probability, we have
\[
\left|\frac{1}{m}\left|N(v) \cap V_\star\right| - \left|N(v) \cap V_{\varphi(v)}\right| \right|
\le
\frac{1}{20}\threshold_{\star} + \frac{3}{400}\threshold_{\star} \le \frac{1}{16}\threshold_{\star}.
\]
By the union bound, this bound holds for all vertices in $V_\star$ simultaneously---and hence $\EE$ occurs---with probability at least $1 - n \cdot n^{-12} = 1 - n^{-11}$.

If $\EE$ does not occur, we can bound both $\sum_{i \in [\machines]}\left|H''_{i}\right|$ and $\sum_{i \in [\machines]}\left|F''_{i}\right|$ by $n$. This contributes at most $n^{-11}\cdot n = n^{-10}$ to the expected size of each of these quantities. Suppose now that $\EE$ occurs. Consider an arbitrary $v \in H''_{i}$ for some $i$.
The number of neighbors of $v$ in $V_i$ lies in the range $\left[\frac{1}{8}\threshold_\star - \frac{1}{16}\threshold_\star,
\frac{3}{2}\threshold_\star + \frac{1}{16}\threshold_\star\right] \subseteq \left[\frac{1}{16}\threshold_\star,2\threshold_\star\right]$.
Moreover, the expected number of vertices $w \in F''_{i}$ that select $v$ in $w_\star$ in Line~\ref{line:randomized_matching:v_star} of \matchHeavy{} is bounded by $2\threshold_\star \cdot \frac{1}{4\threshold_\star} = \frac{1}{2}$. It follows that $\Exp\left[\left|F''_{i}\right|\right] \le \frac{1}{2}\Exp\left[|H''_{i}|\right]$, given $\EE$. We now lower bound the expected size of $\tM_i$ given $\EE$.
What is the probability that some vertex $w \in F_{i}$ selects $v$ as $w_\star$ in \matchHeavy{} and $(v,w)$ is added to $\tM_i$?

This occurs if one of $v$'s neighbors $w$ is added to $F_{i}$ and selects $v$ as $w_\star$, and additionally, $v$ and $w$ are colored blue and red, respectively.
The number of $v$'s neighbors is at least $\frac{1}{16}\threshold_\star$. Since each vertex $w$ in $V_i$ has at most $2\threshold_\star$ neighbors, the number of heavy neighbors of $w$ is bounded by the same number.
This implies that in the process of selecting $F_{i}$, only the first branch in the definition of $\prF$ is used and each vertex $w$ is included with probability exactly equal to the number of its neighbors in $H_{i}$ divided by $4\threshold_{t+1}$.
Then each heavy neighbor of $w$ is selected as $w_\star$ with probability one over the number of heavy neighbors of $w$.
What this implies is that each neighbor $w$ of $v$ is selected for $F_{i}$ \emph{and} selects $v$ as $w_\star$ with probability exactly $(4\threshold_\star)^{-1}$.
Hence the probability that $v$ is \emph{not} selected as $w_\star$ by any of its at least $\frac{1}{16}\threshold_\star$ neighbors $w$ can be bounded by
\[
 \left(1 - \frac{1}{4\threshold_\star}\right)^{\frac{1}{16}\threshold_\star}
 \le
 \exp\left(-\frac{1}{4\threshold_\star}  \cdot \frac{1}{16} \threshold_\star\right)
 = e^{-1/64}.
\]
Therefore the probability that $v$ is selected by some vertex $w \in F_{i}$ as $w_\star$ is at least
$1-e^{-1/64} \ge 1/100$.
Then with probability $1/4$, these two vertices have appropriate colors and this or another edge incident to $v$ with the same properties is added to $\tM_{i}$.
In summary, the probability that an edge $(v,w)$ for some $w$ as described is added to $\tM_{i}$ is at least $1/400$. Since we do not count any edge in the matching twice for two heavy vertices, by the linearity of expectation $\Exp\left[\left| \tM_{i} \right|\right] \ge \frac{1}{400}\Exp\left[|H''_i|\right]$ given $\EE$. Overall, given $\EE$, we have
\[
\sum_{i \in [\machines]}
\Exp\left[
 \left|H''_{i}\right| +
 \left|F''_{i}\right|
 \right] \le
\frac{3}{2}\sum_{i \in [\machines]}
 \Exp\left[\left|H''_{i}\right|\right]
 \le
600\sum_{i \in [\machines]}
 \Exp\left[\left|\tM_{i}\right|\right].
\]
In general, without conditioning on $\EE$,
\[
\sum_{i \in [\machines]}
 \Exp\left[
 \left|H''_{i}\right| +
 \left|F''_{i}\right|
 \right] \le 2\cdot n^{-10}+
600\sum_{i \in [\machines]}
\Exp\left[
 \left|\tM_{i}\right|
 \right].
\]

We now combine bounds on all terms to finish the proof of the lemma.
\begin{align*}
 \sum_{i \in [\machines]}\Exp\left[\left|H_{i}\cup F_{i}\right|\right]
 & \le \sum_{i \in [\machines]}
 \Exp\left[
 \left|H'_{i}\right| +
 \left|F'_{i}\right| +
 \left|H''_{i}\right| +
 \left|F''_{i}\right|
 \right]
 \\
 & \le n^{-11} + n^{-10} + 2n^{-10} + 600
 \sum_{i \in [\machines]}\Exp\left[\left|\tM_{i}\right|\right]
 \\
 & \le n^{-9} + 1200\sum_{i \in [\machines]}\Exp\left[\left|M_{i}\right|\right].\qedhere
\end{align*}
\end{proof}
\undefine{\curG}

\subsection{Independence}
Next we prove Lemma~\ref{lem:independence}.
We start with an auxiliary lemma that gives a simple criterion under which an independent distribution remains independent after conditioning on a random event. Consider a random vector with independently distributed coordinates. Suppose that for any value of the vector, a random event $\EE$ occurs when all coordinates ``cooperate'', where each coordinate cooperates independently with probability that depends only on the value of that coordinate. We then show that the distribution of the vector's coordinates given $\EE$ remains independent.
\begin{lemma}\label{lem:cond_independence}
Let $k$ be a positive integer and $A$ an arbitrary finite set.
Let $X = (X_1,\ldots,X_k)$ be a random vector in $A^k$ with independently distributed coordinates. Let $\EE$ be a random event of non-zero probability. If there exist functions $p_i:A \to [0,1]$, for $i \in [k]$, such that for any $x=(x_1,\ldots,x_k) \in A^k$ appearing with non-zero probability,
\[\Pr[\EE | X = x] = \prod_{i=1}^k p_i(x_i),\]
then the conditional distribution of coordinates in $X$ given $\EE$ is independent as well.
\end{lemma}

\begin{proof}
Since the distribution of coordinates in $X$ is independent, there are $k$ probability mass functions $p'_i:A \to [0,1]$, $i \in [k]$, such that for every $x = (x_1,\ldots,x_k) \in A^k$, $\Pr[X = x] = \prod_{i=1}^{k} p'_i(x_i)$. The probability of $\EE$ can be expressed as
\begin{align*}
 \Pr[\EE] &= \sum_{x=(x_1,\ldots,x_k) \in A^k} \Pr[\EE \land X=x] =  \sum_{\substack{x=(x_1,\ldots,x_k) \in A^k \\ \Pr[X = x] > 0}} \Pr[\EE|X=x] \cdot \Pr[X=x]\\
 &= \sum_{x=(x_1,\ldots,x_k) \in A^k} \prod_{i=1}^k p_i(x_i) p'_i(x_i) = \prod_{i=1}^k \sum_{y \in A} p_i(y)p'_i(y).
\end{align*}
Note that since the probability of $\EE$ is positive, each multiplicative term $\sum_{y \in A} p_i(y)p'_i(y)$, $i \in [k]$, in the above expression is positive.
We can express the probability of any vector $x=(x_1,\ldots,x_k) \in A^k$ given $\EE$ as follows:
\begin{align*}
 \Pr[X = x | \EE] &= \frac{\Pr[\EE \land X=x]}{\Pr[\EE]}= \frac{\Pr[\EE|X=x] \cdot \Pr[X=x]}{\Pr[\EE]}\\
 &= \frac{\prod_{i=1}^k p_i(x_i) p'_i(x_i)}{\prod_{i=1}^k \sum_{y \in A} p_i(y)p'_i(y)} = \prod_{i=1}^k \frac{p_i(x_i) p'_i(x_i)}{\sum_{y \in A} p_i(y)p'_i(y)}.
\end{align*}
We define $p''_i:A \to [0,1]$ as $p''_i(x) \eqdef p_i(x_i) p'_i(x_i) / \sum_{y \in A} p_i(y)p'_i(y)$ for each $i \in [k]$. Each $p''_i$ is a valid probability mass function on $A$. As a result we have $\Pr[X = x | \EE] = \prod_{i=1}^k p''_i(x_i)$, which proves that the distribution of coordinates in $X$ given $\EE$ is still independent with each coordinate distributed according to its probability mass function $p''_i$.
\end{proof}
We now prove Lemma~\ref{lem:independence} by applying Lemma~\ref{lem:cond_independence} thrice. We refer to functions $p_i$, which describe the probability of each coordinate cooperating, as \emph{cooperation probability functions}.

\lemmaindependence*
\begin{proof}%
$\CC$ can be expressed as
\[\CC = \left(\family{R^\star},\family{H^\star},\family{F^\star}\right)\]
for some subsets $R^\star_{i}$, $H^\star_{i}$, and $F^\star_{i}$ of $V_\star$, where $i \in [\machines]$. We write $\Phi$ to denote the random assignment of vertices to sets selected in Line~\ref{line:full_phase:pick_assignment} of \globalPhase. $\Phi$ is a random variable distributed according to $\DD$.

Let $\EE_R$ be the event that for all $i \in [\machines]$, the reference set $R_{i}$ generated for the $i$-th induced subgraph by \localPhase{} equals exactly $R^\star_{i}$. A vertex $v$ that is assigned to a set $V_i$ is included in $R_{i}$ with probability exactly $\prR$, independently of other vertices. Hence once we fix an assignment $\varphi:V_\star \to [\machines]$ of vertices to sets $V_i$, we can express the probability of $\EE_R$ as a product of probabilities that each vertex cooperates. More formally, $\Pr[\EE_R | \Phi = \varphi] = \prod_{v \in V_\star} q_v(\varphi(v))$ for cooperation probability functions $q_v:[\machines] \to [0,1]$ defined as follows.
\begin{itemize}
  \item If $v \in \bigcup_{i \in [\machines]} R^\star_{i}$, there is exactly one $i \in [m]$ such that $v \in R^\star_{i}$. If $v$ is not assigned to $V_i$, $\EE_R$ cannot occur. If it is, $v$ cooperates with $\EE_R$ with probability exactly $\prR$, i.e., the probability of the selection for $R_{i}$. For this kind of $v$, the cooperation probability function is
  \[q_v(i) \eqdef \begin{cases}\prR&\mbox{if $v \in R^\star_{i}$,}\\0&\mbox{if $v \not\in R^\star_{i}$.}\end{cases}\]

  \item If $v \not\in \bigcup_{i \in [\machines]} R^\star_{i}$, $v$ cooperates with $\EE_R$ if it is not selected for $R_{\varphi(v)}$, independently of its assignment $\varphi(v)$, which happens with probability exactly $1-\prR$. Therefore, the cooperation probability can be defined as $q_v(i) \eqdef 1-\prR$ for all $i \in [\machines]$.
\end{itemize}
We invoke Lemma~\ref{lem:cond_independence} to conclude that the conditional distribution of values of $\Phi$ given $\EE_R$ is independent as well.

We now define an event $\EE_H$ that both $\EE_R$ occurs and for all $i \in [\machines]$, $H_{i}$, the set of heavy vertices constructed for the $i$-th subgraph equals exactly $H^\star_{i}$. We want to show that the conditional distribution of values of $\Phi$ given $\EE_H$ is independent. Note that if $\Phi$ is selected from the conditional distribution given $\EE_R$ (i.e., all sets $R_i$ are as expected) and we fix the assignment $\phi:V_\star \to [\machines]$ of vertices to sets $V_i$, then each vertex $v \in V_\star$ is assigned to $H_{\phi(v)}$---this the only set $H_i$ to which it can be assigned---independently of other vertices. As a result, we can express the probability of $\EE_H$ given $\EE_R$ and $\varphi$ being the assignment as a product of cooperation probabilities for each vertex. More precisely,
$\Pr[\EE_H | \Phi = \varphi, \EE_R] = \prod_{v \in V_\star} q'_v(\varphi(v))$ for cooperation probability functions $q'_v:[\machines] \to [0,1]$
defined as follows, where $\threshold_\star$ is the threshold used in the $\machines$ executions of \localPhase.
\begin{itemize}

  \item If $v \in \bigcup_{i \in [\machines]} H^\star_{i}$, then there is exactly one $i$ such that $v \in H^\star_{i}$. $\EE_H$ can only occur if $v$ is included in the corresponding $H_{i}$. This cannot happen if $v$ is not assigned to the corresponding $V_{i}$ by $\varphi$. If $v$ is assigned to this $V_i$, it has to be selected for $H_i$, which happens with probability $\prH\left( |N(v) \cap R^\star_i| / (\prR \threshold_\star) \right)$.
  The cooperation probability function can be written in this case as
  \[q'_v(i) \eqdef \begin{cases}
  \prH(|N(v)\cap R^\star_i|/(\prR \threshold_\star))&\mbox{if $v \in H^\star_i$,}\\
   0&\mbox{if $v \not\in H^\star_i$.}\\
  \end{cases}\]

  \item If $v \not\in \bigcup_{i \in [\machines]} H^\star_{i}$, $v$ cannot be included in $H_{i}$ corresponding to the set $V_i$ to which it is assigned for $\EE_H$ to occur. This happens with probability $1-\prH(|N(v)\cap R^\star_i|/(\prR \threshold_\star))$. Hence, we can define
  $q'_v(i) \eqdef 1-\prH(|N(v)\cap R^\star_i|/(\prR \threshold_\star))$ for all $i \in [\machines]$.

 \end{itemize}
We can now invoke Lemma~\ref{lem:cond_independence} to conclude that the distribution of values of $\Phi$ given $\EE_H$ is independent.

Finally, we define $\EE_F$ to be the event that both $\EE_H$ occurs and for each $i \in [\machines]$, $F_{i}$, the set of friends selected for the $i$-th induced subgraph, equals exactly $F^\star_{i}$. We observe that once $\Phi$ is fixed to a specific assignment $\varphi:V_\star \to [\machines]$ and $\EE_H$ occurs (i.e., all sets $R_i$ and $H_i$ are as in $\CC$), then each vertex is independently included in $F_{\varphi(v)}$ with some specific probability that depends only on $H_{\varphi(v)}$, which is already fixed. In this setting, we can therefore express the probability of $\EE_F$, which exactly specifies the composition of sets $F_i$, as a product of values provided by some cooperation probability functions $q''_v : [\machines] \to [0,1]$. More precisely, $\Pr[\EE_F | \Phi = \varphi, \EE_H] = \prod_{v \in V_\star} q''_v(\varphi(v))$ for $q''_v$ that we define next.
\begin{itemize}
  \item If $v \in \bigcup_{i \in [\machines]} F^\star_{i}$, then there is exactly one $i$ such that $v \in F^\star_{i}$. $\EE_F$ cannot occur if $v$ is not assigned to $V_i$ and selected for $F_{i}$. Hence, the cooperation probability function for $v$ is
  \[q''_v(i) \eqdef \begin{cases}
  \prF(|N(v)\cap H^\star_i|/\threshold_\star)&\mbox{if $v \in F^\star_i$,}\\
  0&\mbox{if $v \not\in F^\star_i$.}
  \end{cases}\]
\item If $v \not \in \bigcup_{i \in [\machines]} F^\star_{i}$, to whichever set $V_i$ vertex $v$ is assigned, it should not be included in $F_{i}$ in order for $\EE_F$ to occur. Hence,
$q''_v(i) \eqdef 1 - \prF(|N(v)\cap H^\star_{i_\star,t}|/\threshold_t)$.
\end{itemize}
We invoke Lemma~\ref{lem:cond_independence} to conclude that the distribution of values of $\Phi$ given $\EE_F$ is independent as well. This is a distribution on assignments for the entire set $V_\star$. If we restrict it to assignments of $V'_\star \subseteq V_\star$, we obtain a distribution that first, is independent as well, and second, equals exactly $\DD[\CC]$.
\end{proof}

\subsection{Near Uniformity}
\label{sec:near_uniform}
In this section we prove Lemma~\ref{lem:preserving_near_uniformity}. We begin by showing a useful property of $\prH$ (see Table~\ref{table:parameters} for definition). Recall that \globalAlgorithm{} selects $H$, the set of heavy vertices, by taking all vertices of degree at least $\threshold/2$. In \localPhase{} the degree estimate of each vertex depends on the number of neighbors in the reference set in the vertex's induced subgraph. We want the decision taken for each vertex to be approximately the same, independently of which subgraph it is assigned to. Therefore, we use $\prH$---which specifies the probability of the inclusion in the set of heavy vertices---which is relatively insensitive to small argument changes. The next lemma proves that this is indeed the case. Small additive changes to the parameter $x$ to $\prH$ have small multiplicative impact on both $\prH(x)$ and $1-\prH(x)$.

\begin{lemma}[Insensitivity of $\prH$]\label{lem:prop_of_prH}
Let $\delta \in [0,(\prHmult/2)^{-1}] = [0,(48 \ln n)^{-1}]$.
For any pair $x$ and $x'$ of real numbers such that $|x - x'| \le \delta$,
$$\prH(x') \in \range{\prH(x)(1\pm\prHmult\delta)}$$
and
$$1-\prH(x') \in \range{(1-\prH(x))(1\pm\prHmult\delta)}.$$
\end{lemma}

\begin{proof}
We define an auxiliary function $f:\mathbb R \to [0,1]$:
  $$f(r) \eqdef
  \begin{cases}
  \frac{1}{2}\exp\left(\frac{\prHmult}{2}r\right)&
  \text{if $r \le 0$,}\\
  1-\frac{1}{2}\exp\left(-\frac{\prHmult}{2} r\right)&
  \text{if $r > 0$.}
  \end{cases}$$
It is easy to verify that for all $r \in \mathbb R$, $\prH(r) = f(r-1/2)$ and $1-\prH(r) = f(-(r-1/2))$.
Therefore, in order to prove the lemma, it suffices to prove that for any $r$ and $r'$ such that $|r-r'|\le \delta$,
\begin{equation}\label{eq:mult_change}
f(r)(1-\prHmult\delta) \le f(r') \le f(r)(1+\prHmult\delta),
\end{equation}
i.e., a small additive change to the argument of $f$ has a limited multiplicative impact on the value of $f$.

Note that $f$ is differentiable in both $(-\infty,0)$ and $(0,\infty)$. Additionally, it is continuous in the entire range---the left and right branch of the function meet at $0$---and both the left and right derivatives at $0$ are equal. This implies that it is differentiable at $0$ as well. Its derivative is
  $$f'(r) =
  \begin{cases}
  \frac{\prHmult}{4} \cdot \exp\left(\frac{\prHmult}{2}r\right)&
  \text{if $r \le 0$,}\\
  \frac{\prHmult}{4} \cdot \exp\left(-\frac{\prHmult}{2}r\right)&
  \text{if $r > 0$,}
  \end{cases}$$
which is positive for all $r$, and therefore, $f$ is strictly increasing.
Note that $f'$ is increasing in $(-\infty,0]$ and decreasing in $[0,\infty)$. Hence the global maximum of $f'$ equals $f'(0) = \prHmult/4$.

In order to prove Inequality \ref{eq:mult_change} for all $r$ and $r'$ such that $|r-r'|\le\delta$, we consider two cases. Suppose first that $r \ge 0$.
By the upper bound on the derivative of $f$,
$$f(r) - \frac{\prHmult}{4} \cdot |r - r'| \le f(r') \le f(r) + \frac{\prHmult}{4} \cdot |r - r'|.$$
Since $r\ge 0$, $f(r) \ge 1/2$. This leads to
$$f(r) - f(r) \cdot \frac{\prHmult}{2} \cdot |r - r'| \le f(r') \le f(r) + f(r) \cdot \frac{\prHmult}{2} \cdot |r - r'|.$$
By the bound on $|r - r'|$,
$$f(r) (1 - \prHmult \delta) \le f(r') \le f(r) (1 + \prHmult \delta),$$
which finishes the proof in the first case.

Suppose now that $r < 0$. Since $f$ is increasing, it suffices to bound the value of $f$ from below at $r-\delta$ and from above and at $r+\delta$. For $r-\delta$, we obtain
\begin{align*}
f(r-\delta) & = \frac{1}{2} \exp\left(\frac{\prHmult}{2} (r-\delta) \right) = f(r) \exp\left(-\frac{\prHmult}{2} \delta\right)\\
                & \ge f(r) \left(1 - \frac{\alpha}{2} \delta\right) \ge f(r) (1 - \prHmult \delta).
\end{align*}
For $r + \delta$, let us first define a function $g:\mathbb R \to \mathbb R$ as
$$g(y) \eqdef \frac{1}{2}\exp\left(\frac{\prHmult}{2}y\right).$$
	For $y \le 0$, $f(y) = g(y)$. For $y > 0$, $g'(y) \geq f'(y)$ and hence, for any $y \in \mathbb R$, $g(y) \ge f(y)$.
As a result, we obtain
$$f(r+\delta) \le g(r+\delta) = \frac{1}{2} \exp\left(\frac{\alpha}{2} (r+\delta) \right)
=f(r) \cdot \exp\left(\frac{\prHmult}{2} \delta \right).$$
By the bound on $\delta$ in the lemma statement, $\frac{\prHmult}{2}\delta \le 1$. It follows from the convexity of the exponential function that for any $y \in [0,1]$, $\exp(y) \le y\cdot \exp(1) + (1-y)\cdot\exp(0) \le 3y + (1 - y) = 1 + 2y$. Continuing the reasoning,
$$f(r+\delta) \le f(r) \cdot \left(1 + 2 \cdot \frac{\prHmult}{2}\delta\right) = f(r) (1 + \prHmult\delta),$$
	which finishes the proof of Inequality~(\ref{eq:mult_change}).
\end{proof}

The main result of this section is Lemma~\ref{lem:preserving_near_uniformity} that states that if a distribution $\DD$ of vertex assignments is near uniform, then \globalPhase{} constructs a configuration $\CC$ such that $\DD[\CC]$ is near uniform as well, and also, the maximum degree in the graph induced by the vertices not removed by \globalPhase{} is bounded.

\lemmauniformity*

\paragraph{Proof overview (of Lemma \ref{lem:preserving_near_uniformity}).} This is the most intricate proof of the entire paper. We therefore provide a short overview. First, we list again the variables in \globalPhase{} and \localPhase{} to which we refer in the proof and define additional convenient symbols. Then we introduce five simple random events (\allEvents) that capture properties needed to prove Lemma~\ref{lem:preserving_near_uniformity}. In Claim~\ref{claim:random_events}, we show that the probability of all these events occurring simultaneously is high. The proof of the claim follows mostly from a repetitive application of the Chernoff bound. In the next claim, Claim~\ref{claim:good_consequences}, we show that the occurrence of all the events has a few helpful consequences. First, high degree vertices get removed in the execution of \globalPhase{} (which is one of our final desired properties).
Second, each vertex $v$ that is not fixed in $\CC$ has a very similar number of neighbors in all sets $R_i$ and it has a very similar number of neighbors in all sets $H_i$. In the final proof of Lemma~\ref{lem:preserving_near_uniformity}, we use the fact that this implies that to whichever set $V_i$ vertex $v$ was assigned in \globalPhase{}, the probability of its removal in \globalPhase{} was more or less the same. This leads to the conclusion that if $v$ was distributed nearly uniformly in $\DD$, it is distributed only slightly less uniformly in $\DD[\CC]$.

\paragraph{Notation.}
To simplify the presentation, for the rest of Section~\ref{sec:near_uniform}, we assume that $\threshold$, $G_\star=(V_\star,E_\star)$, $\machines$, and $\DD$ are the parameters to \globalPhase{} as in the statement of Lemma~\ref{lem:preserving_near_uniformity}. Additionally, for each $i \in [\machines]$, $R_i$, $H_i$, and $F_i$ are the sets constructed by \localPhase{} for the $i$-th subgraph in the execution of \globalPhase.
We also write $\CC$ to denote the corresponding $\machines$-configuration, i.e., $\CC = \left(\family{R},\family{H},\family{F}\right)$.
Furthermore, for each $v \in V_\star$, $\degEst{v}$ is the estimate of $v$'s degree in the subgraph to which it was assigned. This estimate is computed in Line~\ref{line:local_phase:estimate_degrees} of \localPhase. We also use $\threshold_\star$ to denote the rescaled threshold passed in all calls to \localPhase, i.e., $\threshold_\star = \frac{\threshold}{\machines}$.

We also introduce additional notation, not present in \globalPhase{} or \localPhase. For each $v \in V_\star$, $d_v \eqdef |N(v) \cap V_\star|$, i.e., $d_v$ is the degree of $v$ in $G_\star$. For each vertex $v \in V_\star$, we also introduce a notion of its \emph{weight}: $w_v \eqdef \prH(d_v/\threshold)$, which can be seen as a very rough approximation of $v$'s probability of being selected for the set of heavy vertices. For any $v \in V_\star$ and $U\subseteq V_\star$, we also introduce notation for the total weight of $v$'s neighbors in $U$:
\[W_v(U) \eqdef \sum_{u \in N(v) \cap U} w_u.\]
Finally, for all $i \in [\machines]$ and $v \in V_\star$, we also introduce a slightly less intuitive notion of the expected number of heavy neighbors of $v$ in the $i$-th subgraph after the degree estimates are fixed in Line~\ref{line:local_phase:estimate_degrees} of \localPhase{} and before vertices are assigned to the heavy set in Line~\ref{line:local_phase:create_H}:
\[h_{v,i} \eqdef \sum_{u \in N(v)\cap V_{i}} \prH\left(\degEst{u} / \threshold_\star\right).\]
Obviously, each $h_{v,i}$ is a random variable.

\paragraph{Convenient random events.}
We now list five random events that we hope all to occur simultaneously with high probability. The first event intuitively is the event that high-degree vertices are likely to be included in the set of heavy vertices in Line~\ref{line:local_phase:create_H} of \localPhase.
\begin{NewEvent}{prH_high}\label{first_event}
For each vertex $v \in V_\star$ such that $d_v\ge \frac{3}{4}\threshold$,
\[\prH\left(\degEst{v}/\threshold_\star\right) \ge 1-\frac{1}{2}n^{-6}.\]
\end{NewEvent}
Another way to define this event would be to state that $\degEst{v}$ for such vertices $v$ is high, but this form is more suitable for our applications later. The next event is the event that all such vertices are in fact classified as heavy.
\begin{NewEvent}{high_heavy}
Each vertex $v \in V_\star$ such that $d_v \ge \frac{3}{4}\threshold$ belongs to $\bigcup_{i \in [\machines]} H_{i}$.
\end{NewEvent}
The next event is the event that low-degree vertices have a number of neighbors in each set $R_i$ close to the mean. This implies that if we were able to move a low-degree vertex $v$ to $V_i$, for any $i \in [\machines]$, its estimated degree $\degEst{v}$ would not change significantly.
\begin{NewEvent}{good_deg_estimate}For each vertex $v \in V_\star$ such that $d_v < \frac{3}{4}\threshold$ and each $i \in [\machines]$, \[\left|\frac{1}{\prR}\left| N(v) \cap R_{i} \right| - \frac{d_v}{\machines}\right| \le \threshold_\star^{3/4} + \frac{3}{4}\eps\threshold_\star.\]
\end{NewEvent}
As a reminder, we use $W_v(U)$ to denote the expected number of vertices in $N(v) \cap U$ that are selected as heavy, where every vertex $u$ is selected with respect to its global degree $d_u$. The next event shows that $W_v(V_{i})$ does not deviate much from its mean.
\begin{NewEvent}{good_weight}For each vertex $v \in V_\star$ such that $d_v < \frac{3}{4}\threshold$ and each $i \in [ \machines ]$, \[\left|W_v(V_{i}) - W_v(V_\star)/m \right| \le \threshold_{\star}^{3/4} + \frac{3}{4}\eps\threshold_{\star}.\]
\end{NewEvent}
Recall that $h_{v,i}$ intuitively expresses the expected number of $v$'s neighbors in the $i$-th induced subgraph at some specific stage in the execution of \localPhase{} for the $i$-th induced subgraph. The final event is the event that for all bounded $h_{v,i}$, the actual number of $v$'s neighbors in $H_i$ does not deviate significantly from $h_{v,i}$.
\begin{NewEvent}{expected_heavy_neighbors}\label{last_event}
For each vertex $v \in V_\star$ and each $i \in [ \machines ]$, if $h_{v,i} \le 2\threshold_\star$, then
\[\left| \left|N(v) \cap H_{i}\right| - h_{v,i} \right| \le \threshold_\star^{3/4}.\]
\end{NewEvent}

\paragraph{High probability of the random events.}
We now show that the probability of all the events occurring is high. The proof follows mostly via elementary applications of the Chernoff bound.
\begin{claim}\label{claim:random_events}
If $\eps \in [0,1/100]$ and $\frac{\threshold}{\machines} \ge 4000\prR^{-2}\ln^2n $, then
\allEvents{} occur simultaneously with probability at least $1-n^{-4}$.
\end{claim}

\begin{proof}
We consider all events in order and later show by the union bound that all of them hold simultaneously with high probability.
In the proof of the lemma, we extensively use the fact that $\threshold_\star = \frac{\threshold}{\machines} \ge 4000\prR^{-2}\ln^2n = 4\cdot 10^{15} \cdot \ln^4 n$.

First, we consider \event{prH_high} and \event{high_heavy}, which we handle together. Consider a vertex $v$ such that $d_v \ge \frac{3}{4}\threshold$. Let $i_\star$ be the index of the set to which it is assigned. Since $\DD$ is $\eps$-near uniform, the expectation of $|N(v) \cap R_{i_\star}|$, the number of $v$'s neighbors in $R_{i_\star}$, is at least $(1-\eps)\frac{3}{4}\prR\frac{\threshold}{\machines} \ge \frac{297}{400}\prR\threshold_\star$.
Since vertices are both assigned to machines independently and included in the reference set independently as well, we can apply Lemma~\ref{lem:chernoff_consequence} to bound the deviation with high probability. The probability that the number of neighbors is smaller than $\frac{9}{10} \cdot \frac{297}{400}\prR\threshold_\star \ge \frac{5}{8}\prR\threshold_\star$ is at most
\[2\exp\left(- \frac{1}{3} \cdot \left(\frac{1}{10}\right)^2\cdot \frac{297}{400}\prR\threshold_\star \right)
\le 2\exp\left(- \frac{1}{405}\prR\threshold_\star\right)
\le 2n^{-9} \le \frac{1}{2}n^{-6}.
\]
Hence with probability at least $1-\frac{1}{2}n^{-6}$, $\degEst{v} \ge \frac{5}{8} \threshold_\star$ and
$\prH\left(\degEst{v} / \threshold_\star \right) \ge 1-\frac{1}{2}n^{-6}$.
If this is the case, $v$ is not included in the set of heavy vertices in Line~\ref{line:local_phase:create_H} of \localPhase{} with probability at most $\frac{1}{2}n^{-6}$. Therefore, $v$ has the desired value of $\prH\left(\degEst{v} / \threshold_\star \right)$ and belongs to $H_{i_\star}$ with probability at least $1 - n^{-6}$. By the union bound, this occurs for all high degree vertices with probability at least $1-n^{-5}$, in which case both \event{prH_high} and \event{high_heavy} occur.

We now show that \event{good_deg_estimate} occurs with high probability. Let $v$ be an arbitrary vertex such that $d_v < \frac{3}{4}\threshold$ and let $i \in [\machines]$. Let $X_{v,i} \eqdef \left|N(v) \cap R_{i}\right|$. $X_{v,i}$ is a random variable.
Since $\DD$ is $\eps$-near uniform, $\Exp \left[X_{v,i}\right] \in \range{(1\pm\eps)\prR d_v / \machines}$. In particular, due to the bounds on $d_v$ and $\eps$, $E[X_{v,i}] \le \prR \threshold_\star$. Due to the independence, we can use Lemma~\ref{lem:chernoff_consequence} to bound the deviation of $X_{v,i}$ from its expectation. We have
\begin{align*}
\Pr\left(|X_{v,i} - \Exp[X_{v,i}]| > \prR \threshold_\star^{3/4}\right)
&\le 2\exp\left(-\frac{1}{3}\cdot \left(\frac{1}{\threshold_\star^{1/4}}\right)^2 \cdot \prR \threshold_\star\right)\\
& = 2 \exp\left(-\frac{1}{3}\prR \threshold_\star^{1/2}\right)
\le 2n^{-21}.
\end{align*}
Hence with probability $1-2n^{-21}$, we have
\begin{align*}
\left|X_{v,i} - \prR \frac{d_v}{\machines}\right|
& \le \left|X_{v,i} - \Exp[X_{v,i}]\right| + \left|\Exp[X_{v,i}] - \prR \frac{d_v}{\machines}\right|
\le \prR \threshold_\star^{3/4} + \eps \prR \frac{d_v}{\machines}\\
& \le \prR \threshold_\star^{3/4} + \frac{3}{4}\eps \prR \threshold_\star.
\end{align*}
By dividing both sides by $\prR$, we obtain the desired bound
\[
\left|\frac{X_{v,i}}{\prR} - \frac{d_v}{\machines}\right|
=
\left|\frac{1}{\prR}\left| N(v) \cap R_{i} \right| - \frac{d_v}{\machines}\right|
\le
\threshold_\star^{3/4} + \frac{3}{4}\eps \threshold_\star.
\]
By the union bound, this holds for all $v$ and $i$ of interest---and therefore, \event{good_deg_estimate} occurs---with probability at least $1 - |V_\star| \cdot \machines \cdot 2n^{-21} \ge 1 - n^{-5}$.

We now move on to \event{good_weight}. Consider a vertex $v$ such that $d_v < \frac{3}{4} \threshold$ and $i \in [\machines]$.
Note that since the weight of every vertex is at most 1,
$W_v(V_\star)/\machines \le d_v/\machines < \frac{3}{4} \threshold_\star$.
Since $\DD[\CC]$ is $\eps$-near uniform, $\Exp\left[W_v(V_{i})\right] \in \range{(1\pm\eps)W_v(V_\star)/\machines}$.
In particular, $\Exp\left[W_v(V_{i})\right] \le \frac{101}{100}W_v(V_\star)/m \le
\frac{101}{100}\cdot\frac{3}{4}\threshold_\star \le \threshold_\star$. Since vertices are assigned to machines independently, we can apply Lemma~\ref{lem:chernoff_consequence} to bound the deviation of $W_v(V_{i})$ from the expectation:
\begin{align*}
\Pr\left(\left|W_v(V_{i}) - \Exp\left[W_v(V_{i})\right]\right| > \threshold_\star^{3/4} \right)
&\le 2\exp\left(- \frac{1}{3} \cdot \left(\frac{1}{\threshold_\star^{1/4}}\right)^2 \cdot \threshold_\star\right)\\
&= 2\exp\left(-\frac{1}{3}\threshold^{1/2}_\star\right)
\le 2n^{-21}.
\end{align*}
As a result, with probability at least $1-2n^{-21}$,
\begin{align*}
\left|W_v(V_{i}) - W_v(V_\star)/\machines\right|
& \le \left|W_v(V_{i}) - \Exp\left[W_v(V_{i})\right]\right| + \left|\Exp\left[W_v(V_{i})\right] - W_v(V_\star)/\machines\right|\\
& \le \threshold_\star^{3/4} + \eps W_v(V_\star)/m \le \threshold_\star^{3/4} + \eps d_v/m \le \threshold_\star^{3/4} + \frac{3}{4}\eps \threshold_\star.
\end{align*}
By the union bound, this holds for all $v$ and $i$ of interest---and therefore, \event{good_weight} occurs---with probability at least $1 - |V_\star|\cdot \machines \cdot 2n^{-21} \ge 1 - n^{-5}$.

To show that \event{expected_heavy_neighbors} occurs with high probability, recall first that $h_{v,i}$ is the expected number of $v$'s neighbors to be added in Line~\ref{line:local_phase:create_H} to $H_{i}$ in the execution of \localPhase{} for the $i$-th subgraph. Note that the decision of adding a vertex to $H_{i}$ is made independently for each neighbor of $v$. Fix a $v \in V_\star$ and $i\in[\machines]$ such that
$h_{v,i} \le 2 \threshold_\star$. We apply Lemma~\ref{lem:chernoff_consequence} to bound the probability of a large deviation from the expectation:
\begin{align*}
\Pr\left(\left| |N(v) \cap H_{i}| - h_{v,i} \right| > \threshold_\star^{3/4}\right)
& \le 2\exp\left(-\frac{1}{3}\cdot\left(\frac{1}{2\threshold_\star^{1/4}}\right)^2 \cdot 2\threshold_\star\right)\\
& = 2\exp\left(-\frac{1}{6}\threshold_\star^{1/2}\right)
\le 2 n^{-10}.
\end{align*}
By the union bound the probability that this bound does not hold for some $v$ and $i$ such that $h_{v,i} \le 2 \threshold_\star$ is by the union bound at most $|V_\star| \cdot \machines \cdot 2n^{-10} \le n^{-5}$. Hence, \event{expected_heavy_neighbors} occurs with probability at least $1 - n^{-5}$.

In summary, \allEvents{}
occur simultaneously with probability at least $1 - 4 \cdot n^{-5} \ge 1 - n^{-4}$ by another application of the union bound.
\end{proof}

\paragraph{Consequences of the random events.} We now show that if all the random events occur, then a few helpful properties hold for every vertex $v$ that is not fixed by the constructed configuration $\CC$. Namely, $v$'s degree is at most $\frac{3}{4}\threshold$, the number of $v$'s neighbors is similar in all sets $R_i$ is approximately the same, and the number of $v$'s neighbors is similar in all sets $H_i$.

\begin{claim}\label{claim:good_consequences}
If \allEvents{} occur for $\eps \in [0,(200 \ln n)^{-1}]$ and $\frac{\threshold}{\machines} \ge 4000\prR^{-2}\ln^2n $, then the following properties hold for every vertex $v \in V_\star$ that is not fixed by $\CC$:
\begin{enumerate}
 \item $d_v < \frac{3}{4}\threshold$.

 \item There exists $\baseR_v$ such that for all $i \in [\machines]$,
 \[\left| N(v) \cap R_{i} \right|/\prR
 \in \range{\baseR_v \pm \left(\threshold_\star^{3/4} + \frac{3}{4}\eps\threshold_\star \right)}.\]

 \item There exists $\baseH_v \in \left[0,\frac{3}{4}\threshold_\star\right]$ such that for all $i \in [\machines]$,
 \[\left| N(v) \cap H_{i} \right| \in \range{\baseH_v \pm \prHmult\left(\threshold_\star^{3/4} + \eps\threshold_\star\right) }.\]
\end{enumerate}
\end{claim}

\begin{proof}We use in the proof of the claim the fact that
$\threshold_\star = \frac{\threshold}{\machines} \ge 4000\prR^{-2}\ln^2n = 4\cdot 10^{15} \cdot \ln^4 n$.
To prove the lemma, we fix a vertex $v$ that is not fixed by $\CC$.
The first property is directly implied by \event{high_heavy}. Suppose that $d_v \ge \frac{3}{4}\threshold$. Then $v$ is included in the $H_{i}$ corresponding to the subgraph to which it has been assigned and $v$ is fixed by $\CC$. We obtain a contradiction that implies that $d_v < \frac{3}{4}\threshold$.

For the second property, we now know that $d_v < \frac{3}{4}\threshold$. The property follows then directly from \event{good_deg_estimate} with $\baseR_v \eqdef d_v/\machines$.

The last property requires a more complicated reasoning. We set $\baseH_v \eqdef W_v(V_\star)/m < \frac{3}{4}\threshold_\star$. Consider any $i \in [m]$. By \event{good_weight},
\begin{equation}
W_v(V_{i}) \in \range{\baseH_v \pm \left(\threshold_\star^{3/4} + \frac{3}{4}\eps\threshold_\star\right)}.
\label{eq:bound_one}
\end{equation}
Consider now an arbitrary $u \in V_\star$. We bound the difference between $w_u = \prH\left(d_u/\threshold\right)$, which can be seen as the ideal probability of the inclusion in the set of heavy vertices, and $\prH\left(\degEst{u}/\threshold_\star\right)$, the actual probability of this event
in Line~\ref{line:local_phase:create_H} of the appropriate execution of \localPhase{}. \newcommand{\diff}{\delta_\star}
Let $\diff \eqdef \prHmult \left( \threshold_{\star}^{-1/4} + \frac{3}{4}\eps \right)$.
We consider two cases.
\begin{itemize}
 \item
If $d_u < \frac{3}{4}\threshold$, by \event{good_deg_estimate}, the monotonicity of $\prH$, and Lemma~\ref{lem:prop_of_prH},
\begin{align*}
\prH\left(\degEst{u}/\threshold_\star\right)
&\in\range{\prH\left(\frac{d_u}{\threshold} \pm \left(\threshold_\star^{-1/4} + \frac{3}{4}\eps\right) \right)}\\
&\subseteq \range{w_u\cdot\left(1 \pm \diff\right)}.
\end{align*}
Note that Lemma~\ref{lem:prop_of_prH} is applied properly because $\threshold_\star^{-1/4} + \frac{3}{4}\eps \le (200 \ln n)^{-1} + (200 \ln n)^{-1} \le (48 \ln n)^{-1}$.

\item If $d_u\ge \frac{3}{4}\threshold$, by \event{prH_high}, $\prH\left(\degEst{u}/\threshold_\star\right) \in \left[1-\frac{1}{2}n^{-6},1\right]$. Concurrently, $w_u \in \left[\prH(3/4),1\right] = \left[1-\frac{1}{2}n^{-12},1\right]$.
Because $\threshold_\star$ is relatively small, i.e., $\threshold_\star \le n$,
\[ \prH\left(\degEst{u}/\threshold_\star\right) \in \range{ w_u \left(1 \pm \threshold_\star^{-1/4}\right)}
\subseteq \range{w_u\cdot\left(1 \pm \diff\right)},
\]
which is the same bound as in the previous case.
\end{itemize}
It follows from the bound that we just obtained and the definitions of $W_v$ and $h_{v,i}$ that
\begin{align}
h_{v,i}
&= \sum_{u \in N(v) \cap V_{i}} \prH\left(\degEst{u} / \threshold_\star\right)
\in \range{\left(1 \pm \diff\right) \cdot \sum_{u \in N(v) \cap V_{i}} w_u}\nonumber\\
& = \range{W_v\left(V_i\right)\cdot \left(1 \pm \diff\right)}.
\label{eq:bound_two}
\end{align}
We now combine bounds (\ref{eq:bound_one}) and (\ref{eq:bound_two}):
\begin{align*}
h_{v,i}
& \in \left[
\baseH_v \left(1 - \diff \right) - \left(\threshold_\star^{3/4} + \frac{3}{4}\eps\threshold_\star\right)\left(1+\diff\right)
,
\baseH_v \left(1 + \diff \right) + \left(\threshold_\star^{3/4} + \frac{3}{4}\eps\threshold_\star\right)\left(1+\diff\right)
\right]\\
&\subseteq \range{\baseH_v \pm \left(\baseH_v\diff + \left(\threshold_\star^{3/4}
+ \frac{3}{4}\eps\threshold_\star\right)\left(1+\diff\right)\right)}.
\end{align*}
Due to the lower bound on $\threshold_\star$, we obtain
$\diff \le \prHmult \left( (200 \ln n)^{-1} + (200 \ln n)^{-1} \right) \le 1$.
This enables us to simplify and further transform the bound on $h_{v,i}$:
\begin{align*}
h_{v,i}
&\in \range{\baseH_v \pm \left(\baseH_v\diff + 2\left(\threshold_\star^{3/4} + \frac{3}{4}\eps\threshold_\star\right)\right)}\\
& \subseteq\range{
\baseH_v \pm
\left(
\frac{3}{4}\prHmult \threshold_\star^{3/4}
+
\frac{9}{16}\prHmult\eps\threshold_\star
+
2\threshold_\star^{3/4}
+
\frac{3}{2}\eps\threshold_\star
\right)
}\\
&\subseteq
\range{
\baseH_v \pm \prHmult\left(\frac{4}{5}\threshold_\star^{3/4} + \eps\threshold_\star\right)
}.
\end{align*}
By applying the bound on $\threshold_\star$ again, we obtain a bound on the magnitude of the second term in the above bound:
\[\prHmult\left(\frac{4}{5}\threshold_\star^{3/4} + \eps\threshold_\star\right)
=
\prHmult \left(\frac{4}{5}\threshold_\star^{-1/4} + \eps\right)\threshold_\star
\le
96 \ln n \left(\frac{1}{200 \ln n} + \frac{1}{200 \ln n}\right)\threshold_\star
\le
\threshold_\star.
\]
This implies that $h_{v,i} \le \baseH_v + \threshold_\star \le 2\threshold_\star$. The condition in \event{expected_heavy_neighbors} holds, and therefore, $\left| \left|N(v) \cap H_{i}\right| - h_{v,i}\right| \le \threshold_\star^{3/4}$. We combine this with the bound on $h_{v,i}$ to obtain
\[
\left|N(v) \cap H_{i}\right|
\in
\range{
\baseH_v \pm \left(\prHmult\frac{4}{5}\threshold_\star^{3/4} + \prHmult\eps\threshold_\star + \threshold_\star^{3/4}\right)
}
\subseteq
\range{
\baseH_v \pm \prHmult\left(\threshold_\star^{3/4} + \eps\threshold_\star\right)
}.\qedhere
\]
\end{proof}%
\undefine{\diff}

\paragraph{Wrapping up the proof of near uniformity.}
We now finally prove Lemma~\ref{lem:preserving_near_uniformity}. Recall that it states
that an $\eps$-near uniform $\DD$ is very likely to result in a near uniform $\DD[\CC]$ with a slightly worse parameter and that all vertices not fixed by $\CC$ have bounded degree. The proof combines the last two claims: Claim~\ref{claim:random_events} and Claim~\ref{claim:good_consequences}. We learn that $\CC$, the $\machines$-configuration constructed in the process is very likely to have the properties listed in Claim~\ref{claim:good_consequences}. One of those properties is exactly the property that all vertices not fixed by $\CC$ have bounded degree. Hence we have to prove only the near uniformity property.
We accomplish this by observing that the probability of $\CC$ equal to a specific $\machines$-configuration $\CC_\star$ with good properties---those in Claim~\ref{claim:good_consequences}---does not depend significantly on to which induced subgraph a given vertex $v$ not fixed in $\CC_\star$ is assigned. This can be used to show that the conditional distribution of $v$ given that $\CC=\CC_\star$ is near uniform as desired.

\begin{proof}[Proof of Lemma~\ref{lem:preserving_near_uniformity}]%
\newcommand{\assign}[1]{\EE_{\rightarrow #1}}%
\newcommand{\nCC}{\widetilde{\CC}}%
\newcommand{\nR}{\widetilde{R}}%
\newcommand{\nH}{\widetilde{H}}%
\newcommand{\nF}{\widetilde{F}}%
By combining Claim~\ref{claim:random_events} and Claim~\ref{claim:good_consequences}, we learn that with probability at least $1-n^{-4}$, all properties listed in the statement of Claim~\ref{claim:good_consequences} hold for $\CC$, the configuration constructed by \globalPhase{}. Since one of the properties is exactly the same as in the statement of Lemma~\ref{lem:preserving_near_uniformity}, it suffices to prove the other one: that $\DD[\CC]$ is
$60\prHmult \left(\threshold_\star^{-1/4}+\eps\right)$-near uniform for $\CC$ with this set of properties.

Fix $\nCC = \left(\family{\nR},\family{\nH},\family{\nF}\right)$ to be an $m$-configuration that has non-zero probability when \globalPhase{} is run for $\DD$ and has the properties specified by Claim~\ref{claim:good_consequences}. Consider an arbitrary vertex $v \in V_\star$. In order to prove the near uniformity of $\DD\left[\nCC\right]$,
we show that $v$ is assigned by it almost uniformly to $[\machines]$. Let $\EE$ be the event that \globalPhase{} constructs $\nCC$, i.e., $\CC = \nCC$. For each $i \in [\machines]$, let $\assign{i}$ be the event that $v$ is assigned to the $i$-th induced subgraph. Let $p:[\machines] \to [0,1]$ be the probability mass function describing the probability of the assignment of $v$ to each of the $\machines$ subgraphs in $\DD$. Obviously, $p(i) = \Pr[\assign{i}]$ for all $i \in [\machines]$. Due to the $\eps$-near uniformity of $\DD$, $p(i) = \range{\frac{1}{m} (1\pm \eps)}$.

For each $i \in [\machines]$, let $q_i \eqdef \Pr[\EE | \assign{i}]$. In order to express all $q_i$'s in a suitable form, we conduct a thought experiment. Suppose $v$ were not present in the graph, but the distribution of all the other vertices in the modified $\DD$ remained the same. Let $q_\star$ be the probability of $\EE$, i.e., $\CC = \nCC$, in this modified scenario.
How does the probability of $\EE$ change if we add $v$ back and condition on its assignment to a machine $i$? Note first that conditioning on $\assign{i}$ does not impact the distribution of the other vertices, because vertices are assigned to machines independently in $\DD$.
In order for $\EE$ still to occur in this scenario, $v$ cannot be assigned to any of $\nR_{i}$, $\nH_{i}$, or $\nF_{i}$, for which it is considered.
Additionally, as long as this the case, $v$ does not impact the behavior of other vertices, which only depends on the content of these sets and independent randomized decisions to include vertices. As a result we can express $q_i$ as a product of $q_\star$ and three probabilities: of $v$ not being included in sets $\nR_{i}$, $\nH_{i}$, or $\nF_{i}$.
\begin{equation}\label{eq:q_i}
q_i
= q_\star \cdot \left(1-\prR\right) \cdot
\left(1-\prH\left(\frac{\left| N(v) \cap \nR_{i} \right|/\prR}{\threshold_\star}\right)\right) \cdot
\left(1-\prF\left(\frac{\left| N(v) \cap \nH_{i}\right|}{\threshold_\star}\right)\right).
\end{equation}
Using the properties listed in Claim~\ref{claim:good_consequences}, we have
\[
\left| N(v) \cap \nR_{i} \right|/\prR
 \in \range{ \baseR_v \pm \left(\threshold_\star^{3/4} + \frac{3}{4}\eps\threshold_\star \right) },
\]
and
\[\left| N(v) \cap \nH_{i} \right| \in \range{ \baseH_v \pm \prHmult\left(\threshold_\star^{3/4} + \eps\threshold_\star\right) },\]
where $\baseR_v$ and $\baseH_v$ are constants independent of machine $i$ to which $v$ has been assigned and $\baseH \le \frac{3}{4}\threshold_\star$. In the next step, we use these bounds to derive bounds on the multiplicative terms in Equation (\ref{eq:q_i}) that may depend on $i$. We also repeatedly use the bounds $\threshold_\star = \frac{\threshold}{\machines} \ge 4000\prR^{-2}\ln^2n = 4 \cdot 10^{15} \cdot \ln^4 n$  and $\eps \le (200 \ln n)^{-1}$ from the lemma statement.
First, due to Lemma~\ref{lem:prop_of_prH},
\begin{align*}
1-\prH\left(\frac{\left| N(v) \cap \nR_{i} \right|/\prR}{\threshold_\star}\right)
& \in \range{1- \prH\left(\frac{\baseR_v}{\threshold_{\star}} \pm \left(\threshold_{\star}^{-1/4} + \frac{3}{4}\eps\right) \right)}
\\
& \subseteq \range{\left(1- \prH\left(\frac{\baseR_v}{\threshold_\star}\right)\right) \cdot \left(1 \pm \prHmult\left(\threshold_\star^{-1/4} + \frac{3}{4}\eps\right) \right)}.
\end{align*}
(Note that the application of Lemma~\ref{lem:prop_of_prH} was correct, because
$\threshold_\star^{-1/4} + \frac{3}{4}\eps \le (200 \ln n)^{-1} + (200 \ln n)^{-1} < (96 \ln n)^{-1}$.)
Second,
\begin{align*}
1-\prF\left(\frac{\left| N(v) \cap \nH_{i}\right|}{\threshold_\star}\right)
& \in \range{ 1 - \prF\left(\frac{\baseH_v}{\threshold_\star} \pm\prHmult\left( \threshold_\star^{-1/4} + \eps \right) \right)}.
\end{align*}
Since $\baseH_v/\threshold_\star \le \frac{3}{4}$ and $\prHmult\left( \threshold_\star^{-1/4} + \eps \right) \le (96 \ln n)
\cdot \left( (200 \ln n)^{-1} + (200 \ln n)^{-1} \right) < 1$, the argument to $\prF$ in the above bound is always less than $4$, and therefore, only one branch of $\prF$'s definitions gets applied. Hence, we can eliminate $\prF$:
\begin{align*}
1-\prF\left(\frac{\left| N(v) \cap \nH_{i}\right|}{\threshold_\star}\right)
& \in \range{ 1 - \frac{\baseH_v}{4\threshold_\star} \pm\frac{\prHmult}{4}\left( \threshold_\star^{-1/4} + \eps \right)}.
\end{align*}
Since $1 - \frac{\baseH_v}{4\threshold_\star} \ge \frac{3}{4}$, we can further transform the bound to
\begin{align*}
1-\prF\left(\frac{\left| N(v) \cap \nH_{i}\right|}{\threshold_\star}\right)
& \in \range{\left(1 - \frac{\baseH_v}{4\threshold_\star} \right) \left(1 \pm \frac{\prHmult}{3}\left( \threshold_\star^{-1/4} + \eps \right) \right)}.
\end{align*}
Let $\delta_1 \eqdef \prHmult\left(\threshold_\star^{-1/4} + \frac{3}{4}\eps\right)$ and $\delta_2 \eqdef \frac{\prHmult}{3}\left( \threshold_\star^{-1/4} + \eps \right)$.
As a result, every $q_i$ can be expressed as $q_i = \eta_v \lambda_i \lambda'_i$, where $\eta_v$ is a constant independent of $i$, $\lambda_i \in \range{1\pm\delta_1}$, and $\lambda'_i \in \range{1\pm\delta_2}$.
For every $i$, we can also write
\[\Pr[\EE \land \assign{i}] = \Pr[\EE | \assign{i}] \cdot \Pr[\assign{i}] = \eta_v \lambda_i \lambda'_i \cdot p(i) = \frac{\eta_v}{\machines} \lambda_i\lambda'_i\lambda''_i, \]
where $\lambda''_i \in \range{1\pm \eps}$. We now express the conditional probability of $v$ being assigned to the $i$-th subgraph in $\DD$ given $\EE$:
\[
\Pr[\assign{i} | \EE]  = \frac{\Pr[\EE \land \assign{i}]}{\sum_{j=1}^{\machines} \Pr[\EE \land \assign{j}]} = \frac{\lambda_i\lambda'_i\lambda''_i}{\sum_{j=1}^{\machines} \lambda_j\lambda'_j\lambda''_j}.
\]
Note that for any $i$, this implies that
\begin{equation}\label{eq:prob_bound}
\frac{1}{\machines} \cdot \frac{(1-\delta_1)(1-\delta_2)(1-\eps)}{(1+\delta_1)(1+\delta_2)(1+\eps)}
\le
 \Pr[\assign{i} | \EE]
\le
\frac{1}{\machines} \cdot \frac{(1+\delta_1)(1+\delta_2)(1+\eps)}{(1-\delta_1)(1-\delta_2)(1-\eps)}.
\end{equation}
Observe that
\[\delta_1 \le (96 \ln n) \cdot \left((7000 \ln n)^{-1} + (250 \ln n)^{-1}\right) < 1/2,\]
and
\[\delta_2 \le \frac{1}{3} \cdot (96 \ln n) \cdot \left((7000 \ln n)^{-1} + (200 \ln n)^{-1}\right) < 1/2.\]
Hence all of $\delta_1$, $\delta_2$, and $\eps$ are at most $1/2$. We can therefore transform (\ref{eq:prob_bound}) to
\[
\frac{1}{\machines} \cdot (1-\delta_1)^2(1-\delta_2)^2(1-\eps)^2
\le
 \Pr[\assign{i} | \EE]
\le
\frac{1}{\machines} \cdot (1+\delta_1)(1+\delta_2)(1+\eps)(1+2\delta_1)(1+2\delta_2)(1+2\eps),
\]
and then
\[
\frac{1}{\machines} \cdot (1-2\delta_1-2\delta_2-2\eps)
\le
 \Pr[\assign{i} | \EE]
\le
\frac{1}{\machines} \cdot (1+45\delta_1+45\delta_2+45\eps).
\]
Hence
\[
\Pr[\assign{i} | \EE]
\in \range{\frac{1}{\machines}\cdot (1 \pm 45(\delta_1+\delta_2+\eps))} \subseteq
\range{\frac{1}{\machines}\cdot \left(1 \pm 60\prHmult \left(\threshold_\star^{-1/4}+\eps\right)\right)},
\]
which finishes the proof that $\DD\left[\nCC\right]$ is $60\prHmult \left(\threshold_\star^{-1/4}+\eps\right)$-near uniform.
\end{proof}%
\undefine{\assign}%
\undefine{\nR}%
\undefine{\nH}%
\undefine{\nF} 
\section{Parallel Algorithm}
\label{sec:parallel}
In this section, we introduce our main parallel algorithm. It builds on the ideas introduced in \globalPhase.
\globalPhase{} randomly partitions the graph into $\machines$ induced subgraphs and runs on each of them \localPhase{}, which resembles a phase of \globalAlgorithm. As we have seen, the algorithm performs well even if vertices are assigned to subgraphs not exactly uniformly so long as the assignment is fully independent. Additionally, with high probability, if we condition on the configuration of sets $R_i$, $H_i$, and $F_i$ that were removed, the distribution of assignments of the remaining vertices is still nearly uniform and also independent.

These properties allow for the main idea behind the final parallel algorithm. We partition vertices randomly into $\machines$ induced subgraphs and then run \localPhase{} multiple times on each of them with no repartitioning in the meantime. In each iteration, for a given subgraph, we halve the local threshold $\threshold_\star$. This corresponds to multiple phases of the original global algorithm. As long as we can show that this approach leads to finding a large matching, the obvious gain is that \emph{multiple} phases of the original algorithm translate to $O(1)$ parallel rounds.
This approach enables our main result: the parallel round complexity reduction from $O(\log n)$ to $O((\log \log n)^2)$.

\newcommand{\setMachines}{$\machines \leftarrow \left\lfloor\sqrt{\frac{n\threshold}{S}}\right\rfloor$}%
\newcommand{\setPhases}{$\phases \leftarrow \left\lceil\frac{1}{16}\log_{120\prHmult}\left(\threshold/\machines\right)\right\rceil$}%
\begin{algorithm}[ht]\algfontsize
  \caption{$\parAlg(G,S)$\algdesc{The final parallel matching algorithm}}\label{alg:parallel}
  \KwIn{\\
  \qquad$\bullet$ graph $G = (V,E)$ on $n$ vertices\\
  \qquad$\bullet$ parameter $S \in \mathbb Z_+$ such that $S \le n$ and $S=n^{\Omega(1)}$ (each machine uses $O(S)$ space)
  }
  \KwOut{matching in $G$}

  \BlankLine

  $\threshold \leftarrow n$, $V' \leftarrow V$, $M \leftarrow \emptyset$

  \While{$\threshold \ge \frac{n}{S}\left(200 \ln n\right)^{32}$\label{line:par_alg:parallel-while-condition}}{

     \tcc{High-probability invariant:~maximum degree in $G[V']$ bounded by $\frac{3}{2}\threshold$}

     \setMachines\phantom{\setPhases}\hspace{-1.5cm}\tcc{number of machines used}
		\label{line:machines-assignment}

     \setPhases\phantom{\setMachines}\hspace{-1.5cm}\tcc{number of phases to emulate}
			\label{line:phases-assignment}

     Partition $V'$ into $\machines$ sets $V_1$, \ldots, $V_\machines$ by assigning each vertex independently uniformly at random.
     \label{line:vertex-partitioning}

     \SetKwBlock{ParallelBlock}{foreach $i \in [\machines]$ do in parallel\label{line:par_alg:start_par_block}}{}
     \ParallelBlock{
         If the number of edges in $G[V_i]$ is greater than $8S$, $V_i \leftarrow \emptyset$.\label{line:check-graph-size}\\
         \lFor{$j \in [\phases]$}{$(V_i,M_{i,j}) \leftarrow \localPhase\left(i, G[V_i], \threshold / \left(2^{j-1} \machines\right)\right)$}\label{line:edge-partitioning}\label{line:par_alg:end_par_block}
     }

     $V' \leftarrow \bigcup_{i=1}^{\machines} V_i$\label{line:V'-update}

     $M \leftarrow M \cup \bigcup_{i=1}^{\machines}\bigcup_{j=1}^{\phases} M_{i,j}$\label{line:M-update}

     $\threshold \leftarrow \threshold / 2^{\phases}$
		\label{line:next-threshold}
  }

  Compute degrees of vertices $V'$ in $G[V']$ and remove from $V'$ vertices of degree at least $2\threshold$.\label{line:remove-high-degree-vertices}

  Directly simulate $M' \leftarrow \globalAlgorithm(G[V'],2\threshold)$,
  using $O(1)$ rounds per phase.\label{line:simulate_global_alg}

  \Return{$M \cup M'$}
\end{algorithm}

We present \parAlg, our parallel algorithm, as Algorithm~\ref{alg:parallel}. We write $S$ to denote a parameter specifying the amount of space per machine. After the initialization of variables, the algorithm enters the main loop in Lines~\ref{line:par_alg:parallel-while-condition}--\ref{line:next-threshold}. The loop is executed as long as $\threshold$, an approximate upper bound on the maximum degree in the remaining graph, is large enough. The loop implements the idea of running multiple iterations of \localPhase{} on each induced subgraph in a random partition. At the beginning of the loop, the algorithm decides on $\machines$, the number of machines, and $\phases$, the number of phases to be emulated. Then it creates a random partition of the current set of vertices that results in $\machines$ induced subgraphs. Next for each subgraph, the algorithm first runs a security check that
the set of edges fits onto a single machine (see Line~\ref{line:check-graph-size}). If it does not, which is highly unlikely, the entire subgraph is removed from the graph. Otherwise, the entire subgraph is sent to a single machine that runs $\phases$ consecutive iterations of \localPhase. Then the vertices not removed in the executions of \localPhase{} are collected for further computation and new matching edges are added to the matching being constructed. During the execution of the loop, the maximum degree in the graph induced by $V'$, the set of vertices to be considered is bounded by $\frac{3}{2}\threshold$ with high probability. Once the loop finishes, we remove from the graph vertices of degree higher than $2\threshold$---there should be none---and we directly simulate \globalAlgorithm{} on the remaining graph, using $O(1)$ rounds per phase.

\subsection{Some Properties of Thresholds}

Before we analyze the behavior of the algorithm, we observe that the value of $\frac{\threshold}{\machines}$ inside the main loop is at least polylogarithmic and that the same property holds for the rescaled threshold that is passed to \localPhase{}.

\begin{lemma}\label{lem:threshold_lb}
Consider a single iteration of the main loop of \parAlg{} (Lines~\ref{line:par_alg:parallel-while-condition}--\ref{line:next-threshold}). Let $\threshold$ and $\machines$ be set as in this iteration. The following two properties hold:
\begin{itemize}
 \item $\threshold/\machines \ge (200 \log n)^{16}$.
 \item The threshold  $\threshold / \left(2^{j-1} \machines\right)$ passed to \localPhase{} in Line~\ref{line:edge-partitioning} is always at least
$(\threshold/\machines)^{15/16} \ge 4000\prR^{-2}\ln^2n $.
\end{itemize}
\end{lemma}

\begin{proof}%
\newcommand{\ratio}{\lambda}%
Let $\phases$ be also as in this iteration of the loop. The smallest threshold passed to \localPhase{} is $\threshold / (2^{\phases-1}\machines)$. Let $\ratio \eqdef S\threshold/n$, where $S$ is the parameter to \parAlg. Due to the condition in Line~\ref{line:par_alg:parallel-while-condition}, $\ratio \ge (200 \ln n)^{32}$.
Note that $\threshold  = \ratio n / S$.
Hence $\machines \le \sqrt{n \threshold / S} = \frac{n}{S} \sqrt{\ratio}$. This implies that $\threshold / \machines \ge \sqrt{\ratio} \ge (200 \ln n)^{16}$, which proves the first claim.
Due to the definition of $\phases$,
\[2^{\phases-1} \le (120 \prHmult)^{\phases - 1} \le (\threshold/\machines)^{1/16}.\]
This implies that
\[\threshold / (2^{\phases-1}\machines) \ge (\threshold/\machines)^{15/16}
\ge (200 \ln n)^{15} > 4\cdot 10^{15} \cdot \ln^4 n = 4000\prR^{-2}\ln^2n.\qedhere\]
\end{proof}

We also observe that the probability of any set of vertices deleted by the security check in Line~\ref{line:check-graph-size} of \parAlg{} is low as long as the maximum degree in the graph induced by the remaining vertices is bounded.
\begin{lemma}\label{lem:no_overflow}
Consider a single iteration of the main loop of \parAlg{} and let $\threshold$ and $V'$ be as in that iteration.
If the maximum degree in $G[V']$ is bounded by $\frac{3}{2}\threshold$, then the probability of any subset of vertices deleted in Line~\ref{line:check-graph-size} is $n^{-8}$.
\end{lemma}

\begin{proof}
Let $\machines$ be as in the same iteration of the main loop of \parAlg.
Consider a single vertex $v \in V'$. The expected number of $v$'s neighbors assigned to the same subgraph is at most $\frac{3}{2} \threshold / \machines$. Recall that due to Lemma~\ref{lem:threshold_lb}, $\frac{\threshold}{\machines} \ge 200 \ln n$. Since the assignment of vertices to machines is fully independent,
by Lemma~\ref{lem:chernoff_consequence} (i.e., the Chernoff bound), the probability that $v$ has more than $2\threshold/\machines$ neighbors is bounded by
\[
2\exp\left(-\frac{1}{3} \cdot \left(\frac{1}{3}\right)^2 \cdot \frac{3}{2}\cdot\frac{\threshold}{\machines}\right)
\le 2 \exp\left(- \frac{1}{18} \cdot 200 \ln n\right) \le n^{-10}.
\]
Therefore, by the union bound, with probability $1-n^{-9}$, no vertex has more than $2\threshold$ neighbors in the same induced subgraph.
As $|V'| \le n$, the expected number of vertices in each set $V_i$ constructed in the iteration of the main loop is at most $n/\machines \ge \threshold/\machines \ge 200 \ln n$. What is the probability that $|V_i| > 2n/\machines$?  By the independence of vertex assignments and Lemma~\ref{lem:chernoff_consequence}, the probability of such event is at most
\[
2\exp\left(- \frac{1}{3} \cdot \frac{n}{\machines} \right) \le 2\exp\left(- \frac{1}{3} \cdot 200 \ln n \right)
\le n^{-10}.
\]
Again by the union bound, the event $|V_i| \le 2 n / \machines$, for all $i$ simultaneously, occurs with probability at least $1-n^{-9}$.
Combining both bounds, with probability at least $1-2n^{-9} \ge 1-n^{-8}$, all induced subgraphs have at most $2n/\machines$ vertices and
the degree of every vertex is bounded by $2\threshold/\machines$. Hence the number of edges in one induced subgraph is at most
$\frac{1}{2} \cdot \frac{2n}{\machines} \cdot \frac{2\threshold}{\machines}= \frac{2n\threshold}{\machines^2}$. By the definition of $\machines$ and the setting of parameters in the algorithm, $\machines \ge \frac{1}{2}\sqrt{\frac{n\threshold}{S}}$, where $S$ is the parameter to \parAlg. This implies that the number of edges is at most $2n\threshold / \left(\frac{1}{2}\sqrt{\frac{n\threshold}{S}}\right)^2 = 8S$ in every induced subgraph with probability $1-n^{-8}$,
in which case no set $V_i$ is deleted in Line~\ref{line:check-graph-size} of \parAlg.
\end{proof}

\subsection{Matching Size Analysis}

The parallel algorithm runs multiple iterations of \localPhase{} on each induced subgraph, without repartitioning. A single iteration on all subgraphs corresponds to running \globalPhase{} once. We now show that in most cases, the global algorithm simulates \globalPhase{} on a well behaved distribution with independently assigned vertices and all vertices distributed nearly uniformly conditioned on the configurations of the previously removed sets $R_i$, $H_i$, and $F_i$. We also show that the maximum degree in the remaining graph is likely to decrease gracefully during the process.

\begin{lemma}\label{lem:good_events}
With probability at least $1-n^{-3}$:
\begin{itemize}
 \item all parallel iterations of \localPhase{} in \parAlg{} on each induced subgraph correspond to running \globalPhase{} on independent and $(200 \ln n)^{-1}$-near uniform distributions of assignments,
 \item the maximum degree of the graph induced by the remaining vertices after the $k$-th simulation of \globalPhase{} is $\frac{3}{2}\threshold / 2^k$.
\end{itemize}
\end{lemma}

\begin{proof}
We first consider a single iteration of the main loop in \parAlg{}. Let $\threshold$, $\phases$, and $\machines$ be set as in this iteration of the loop.
For $j \in [\phases]$, let $\threshold_j \eqdef \threshold/\left(2^{j-1} \machines\right)$ be the threshold passed to \localPhase{} for the $j$-th iteration of \localPhase{} on each of the induced subgraphs. The parallel algorithm assigns vertices to subgraphs and then iteratively runs \localPhase{} on each of them. In this analysis we ignore the exact assignment of vertices to subgraphs until they get removed as a member of one of sets $R_i$, $H_i$, or $F_i$.
Instead we look at the conditional distribution on assignments given the configurations of sets $R_i$, $H_i$, and $F_i$ removed in the previous iterations corresponding to \globalPhase.
We write $\DD_j$, $1 \le j \le \phases$, to denote this distribution of assignments before the execution of $j$-th iteration of \localPhase{} on the induced subgraphs, which corresponds to the $j$-th iteration of \globalPhase{} for this iteration of the main loop of \parAlg{}. Additionally, we write $\DD_{\phases + 1}$ to denote the same distribution after the $\tau$-th iteration, i.e., at the end of the execution of the parallel block in Lines~\ref{line:par_alg:start_par_block}--\ref{line:par_alg:end_par_block} of \parAlg.
Due to Lemma~\ref{lem:independence}, the distributions of assignments are all independent. We define $
\eps_j$, $j \in [\phases+1]$, to be the minimum positive value such that $\DD_j$ is $\eps_j$-near uniform. Obviously, $\eps_1 = 0$, since the first distribution corresponds to a perfectly uniform assignment. We want to apply Lemma~\ref{lem:preserving_near_uniformity} inductively to bound the value of $\eps_{j+1}$ as a function of $\eps_j$ with high probability. The lemma lists two conditions: $\eps_j$ must be at most $(200 \ln n)^{-1}$ and the threshold passed to \globalPhase{} has to be at least $4000\prH^{-2}\ln^2 n$. The latter condition holds due to Lemma~\ref{lem:threshold_lb}. Hence as long as $\eps_j$ is sufficiently small, Lemma~\ref{lem:preserving_near_uniformity} implies that with probability at least $1-n^{-4}$,
\[
\eps_{j+1} \le 60\prHmult\left(\left(\frac{\threshold}{2^{\phases-1}\machines}\right)^{-1/4}+\eps_{j}\right)
\le
60\prHmult\left(\left(\frac{\threshold}{\machines}\right)^{-15/64}+\eps_{j}\right),
\]
and no high degree vertex survives in the remaining graph.
One can easily show by induction that if this recursion is satisfied for all $1 \le j \le \tau$, then $\eps_{j} \le (120 \alpha)^{j-1} \cdot \left(\frac{\threshold}{\machines}\right)^{-15/64}$ for all $j \in [\phases+1]$. In particular, by the definition of $\phases$ and Lemma~\ref{lem:threshold_lb}, for any $j \in [\phases]$,
\[\eps_{j}
\le
(120 \prHmult)^{\tau - 1} \cdot \left(\frac{\threshold}{\machines}\right)^{-15/64}
\le
\left(\frac{\threshold}{\machines}\right)^{1/16} \cdot \left(\frac{\threshold}{\machines}\right)^{-15/64}
\le
\left(\frac{\threshold}{\machines}\right)^{-11/64} \le (200 \ln n)^{-1},
\]
This implies that as long the unlikely events specified in Lemma~\ref{lem:preserving_near_uniformity} do not occur for any phase in any iteration of the main loop of \parAlg, we obtain the desired properties: all intermediate distributions of possible assignments are $(200 \ln n)^{-1}$-near uniform and
the maximum degree in the graph decreases at the expected rate. It remains to bound the probability of those unlikely events occurring for any phase. By the union bound, their total probability is at most $\log n \cdot n^{-4} \le n^{-3}$.
\end{proof}

We now prove that the algorithm finds a large matching with constant probability.

\begin{theorem}
\label{theorem:constant_apx}
Let $M_{\rm OPT}$ be an arbitrary maximum matching in a graph $G$. With $\Omega(1)$ probability,
\parAlg{} constructs a matching of size $\Omega(|M_{\rm OPT}|)$.
\end{theorem}

\begin{proof}
By combining Lemma~\ref{lem:no_overflow} and Lemma~\ref{lem:good_events}, we learn that with probability at least $1- n \cdot n^{-8} - n^{-3} \ge 1 - 2n^{-3}$, we obtain a few useful properties. First, all relevant distributions corresponding to iterations of \globalPhase{}
are independent and $(200 \ln n)^{-1}$-near uniform. Second, the maximum degree in the graph induced by vertices still under consideration is bounded by $\frac{3}{2}\threshold$ before and after every simulated execution of \globalPhase{}, where $\threshold$ is the corresponding. As a result, no vertex is deleted in Lines~\ref{line:check-graph-size} or \ref{line:remove-high-degree-vertices} due to the security checks. 

We now use Lemma~\ref{lem:modified_phase_matching_size} to lower bound the expected size of the matching created in every \globalPhase{} simulation. Let $\phases_\star$ be the number of phases we simulate this way. We have $\phases_\star \le \log n$. Let $\mathbf{H}_j$, $\mathbf{F}_j$, and $\mathbf{M}_j$ be random variables equal to the total size of sets $H_i$, $F_i$, and $M_i$ created in the $j$-th phase. If the corresponding distribution in the $j$-th phase is near uniform and the maximum is bounded, Lemma~\ref{lem:modified_phase_matching_size} yields
\[
\Exp\left[\mathbf{H}_j + \mathbf{F}_j\right] \le n^{-9} + 1200 \cdot \Exp\left[\mathbf{M}_j \right],
\]
i.e., 
\[\Exp\left[\mathbf{M}_j \right] \ge \frac{1}{1200}\left(\Exp\left[\mathbf{H}_j + \mathbf{F}_j\right] -  n^{-9} \right). \]
Overall, without the assumption that the conditions of Lemma~\ref{lem:modified_phase_matching_size} are always met, we obtain a lower bound
\[
\sum_{j\in[\phases_\star]}\Exp\left[\mathbf{M}_j \right] \ge \sum_{j\in[\phases_\star]}
\frac{1}{1200}\left(\Exp\left[\mathbf{H}_j + \mathbf{F}_j\right] -  n^{-9} \right)
- 2n^{-3} \cdot \frac{n}{2},
\]
in which we consider the worst case scenario that we lose as much as $n/2$ edges in the size of the constructed matching when the unlikely negative events happen.
\parAlg{} continues the construction of a matching by directly simulating the global algorithm. Let $\phases'_\star$ be the number of phases in that part of the algorithm. We define 
$\mathbf{H}'_j$, $\mathbf{F}'_j$, and $\mathbf{M}'_j$, for $j \in [\phases'_\star]$, to be random variables equal to the size of sets $H$, $F$, and $\widetilde M$ in \globalAlgorithm{} in the $j$-th phase of the simulation. By Lemma~\ref{lem:global_alg_expectation}, we have
\[
\sum_{j\in[\phases'_\star]}\Exp\left[\mathbf{M}'_j \right] \ge \sum_{j\in[\phases'_\star]}
\frac{1}{50}\left(\Exp\left[\mathbf{H}'_j + \mathbf{F}'_j\right]\right).
\]
By combining both bounds we obtain a lower bound on the size of the constructed matching.
Let 
\[\mathbf{M}_\star \eqdef \sum_{j\in[\phases_\star]}\Exp\left[\mathbf{M}_j \right]
+
\sum_{j\in[\phases'_\star]}\Exp\left[\mathbf{M}'_j \right]
\]
be the expected matching size, and let 
\[\mathbf{V}_\star \eqdef 
\sum_{j\in[\phases_\star]} \Exp\left[\mathbf{H}_j + \mathbf{F}_j\right]
+
\sum_{j\in[\phases'_\star]} \Exp\left[\mathbf{H}'_j + \mathbf{F}'_j\right].
\]
We have
\[
\mathbf{M}_\star \ge \frac{1}{1200}\mathbf{V}_\star - \frac{1}{n^2}.
\]
Consider a maximum matching $M_{\rm OPT}$. At the end of the algorithm, the graph is empty. The expected number of edges in $M_{\rm OPT}$ incident to a vertex in one of the reference sets is bounded by $|M_{\rm OPT}| \cdot 2 \prR \cdot \log n \le 10^{-5}|M_{\rm OPT}|$. The expected number of edges removed by the security checks is bounded by $\frac{n}{2} \cdot n^{-3}$. Hence the expected number of edges in $M_{\rm OPT}$ deleted as incident to vertices that are heavy or are friends is at least $(1-10^{-5})|M_{\rm OPT}| - 1/(2n^2)$. Since we can assume without the loss of generality that the graph is non-empty, it is at least $\frac{1}{2}|M_{\rm OPT}|$. Hence $\mathbf{V}_\star \ge \frac{1}{2}|M_{\rm OPT}|$, and $\mathbf{M}_\star \ge \frac{1}{2400}|M_{\rm OPT}| - \frac{1}{n^2}.$ For sufficiently large $n$ (say, $n \ge 50$), $\mathbf{M}_\star \ge \Omega\left(|M_{\rm OPT}|\right)$ and by an averaging argument, \parAlg{} has to output an $O(1)$-multiplicative approximation to the maximum matching with $\Omega(1)$ probability.
For smaller $n$, it is not difficult to show that at least one edge is output by the algorithm with constant probability as long as it is not empty.
\end{proof}

Finally, we want to argue that the above procedure can be used to compute $2+\epsilon$ approximation to maximum matching at the cost of increasing the 
running time by a factor of $\log(1/\eps)$. The idea is to; execute algorithm $\parAlg{}$ to compute constant approximate matching; remove this matching
from the graph; and repeat. 

\begin{corollary}
\label{corollary:2_apx}
Let $M_{\rm OPT}$ be an arbitrary maximum matching in a graph $G$. 
For any $\eps > 0$,
executing \parAlg{} on $G$ and removing a constructed matching repetitively, $O(\log(1/\eps))$ times, finds a multiplicative $(2+\eps)$-approximation to maximum matching, with $\Omega(1)$ probability.
\end{corollary} 
\begin{proof}
Assume that the \parAlg{} succeeds with probability $p$ and computes $c$-approximate matching. Observe that each successful 
execution of \parAlg{} finds a matching $M_c$ of size at least $\frac{1}{c}|M_{\rm OPT}|$. Removal of $M_c$ from the graph 
decreases the size of optimal matching by at least $\frac{1}{c}|M_{\rm OPT}|$ and at most by $\frac{2}{c}|M_{\rm OPT}|$, because each edge of $M_c$ can be incident to at most
two edges of $M_{\rm OPT}$. Hence, when the size of the remaining matching drops to at most $\eps |M_{\rm OPT}|$, we have an 
$2+\eps$-multiplicative approximation to maximum matching constructed. The number $t$ of successful applications of \parAlg{} need to satisfy.
\[
\left(1-\frac{1}{c}\right)^t \le \eps.
\]
This gives $t=O(\log(1/\eps))$. In $\lceil t/p \rceil=O(\log(1/\eps))$ executions,
we have $t$ successes with probability at least $1/2$ by the properties of the median of the binomial distribution. 
\end{proof}

\section{MPC Implementation Details}
\label{sec:complexity-and-implementation}

In this section we present details of an MPC implementation of our algorithm. We also analyze its round and space complexity.
In the description we heavily use some of the subroutines described in~\cite{goodrich2011sorting}. While the model used there is different, the properties of the distributed model used in~\cite{goodrich2011sorting} also hold in the MPC model. Thus, the results carry over to the MPC model.

The results of~\cite{goodrich2011sorting} allow us to sort a set $A$ of $O(N)$ key-value pairs of size $O(1)$ and for every element of a sorted list, compute its index. Moreover, we can also do a parallel search: given a collection $A$ of $O(N)$ key-value pairs and a collection of $O(N)$ queries, each containing a key of an element of $A$, we can annotate each query with the corresponding key-value pair from $A$. Note that multiple queries may ask for the same key, which is nontrivial to parallelize.
If $S = n^{\Omega(1)}$, all the above operations can be implemented in $O(1)$ rounds.

The search operation allows us to broadcast information from vertices to their incident edges. Namely, we can build a collection of key-value pairs, where each key is a vertex and the value is the corresponding information. Then, each edge $\{u,v\}$ may issue two queries to obtain the information associated with $u$ and $v$.

\newcommand{\discard}{\mbox{\tt discarded}}
\newcommand{\matched}{\mbox{\tt matched}}

\subsection{$\globalAlgorithm$}
\label{sec:implementation-of-centralized-algorithm}

We first show how to implement $\globalAlgorithm$, which is called in Line~\ref{line:simulate_global_alg} of $\parallelAlgorithm$.

\begin{lemma}
\label{lemma:globalAlgorithm-implementation}
Let $S = n^{\Omega(1)}$. There exists an implementation of $\globalAlgorithm$ in the MPC model, which with high probability executes $O(\ln \widetilde\threshold)$ rounds and uses $O(S)$ space per machine.
\end{lemma}
\begin{proof}
We first describe how to solve the following subproblem.
Given a set $X$ of marked vertices, for each vertex $v$ compute $|N(v) \cap X|$.
When all vertices are marked, this just computes the degree of every vertex.

The subproblem can be solved as follows. Create a set $A_X = \{(u, v) \mid u \in V, v \in X, \{u, v\} \in E\} \cup \{(v, -\infty), (v, \infty)\ |\ v \in V\}$, and sort its elements lexicographically.
Denote the sorted sequence by $Q_X$. Then, for each element of $A_X$ compute its index in $Q_A$.

Note that $|N(v) \cap X|$ is equal to the number of elements in $Q_X$ between $(v, -\infty)$ and $(v, \infty)$. Thus, having computed the indices of these two elements, we can compute $|N(v) \cap X|$.

\SetKw{foreach}{foreach}
Let us now describe how to implement \globalAlgorithm{}.
We can compute the degrees of all vertices, as described above.
Once we know the degrees, we can trivially mark the vertices in $H$.
The next step is to compute $F$ and for that we need to obtain $|N(v) \cap H|$,
which can be done as described above.

	After that, \globalAlgorithm{} computes a matching in $G[H \cup F]$ by calling \matchHeavy{} (see Algorithm~\ref{alg:randomized_matching}).
	In the first step, \matchHeavy{} assigns to every $v \in F$ a random neighbor $v_{\star}$ in $H$.
	This can again be easily done by using the sequence $Q_H$ (i.e. $Q_X$ build for $X = H$).
Note that for each $v \in F$ we know the number of neighbors of $v$ that belong to $H$.
Thus, each vertex $v$ can pick an integer $r_v \in [1, |N(v) \cap H|]$ uniformly at random.
Then, by adding $r_v$ and the index of $(v, -\infty)$ in $Q_H$, we obtain the index in $Q_H$, which corresponds to an edge between $v$ and its random neighbor in $H$.
The remaining lines of \matchHeavy{} are straightforward to implement.
The vertices can trivially pick their colors.
After that, the set $E_{\star}$ can be easily computed by transmitting data from vertices to their adjacent edges.
Implementing the following steps of \matchHeavy{} is straightforward.
Finally, picking the edges to be matched is analogous to the step, when for each $v \in F$ we picked a random neighbor in $H$.

Overall, each phase of $\globalAlgorithm$ (that is, iteration of the main loop) is executed in $O(1)$ rounds. Thus, by Lemma~\ref{lem:global_alg_expectation}, $\globalAlgorithm$ can be simulated in $O(\ln \widetilde\threshold)$ rounds as advertised.
\end{proof}

\subsection{Vertex and edge partitioning}
\label{sec:vertex-edge-partitioning}
We now show how to implement Line~\ref{line:vertex-partitioning} and compute the set of edges that are used in each call to \localPhase{} in Line~\ref{line:edge-partitioning} of $\parallelAlgorithm$.
Our goal is to annotate each edge with the machine number it is supposed to go to. To that end, once the vertices pick their machine numbers, we broadcast them to their adjacent edges. Every edge that receives two equal numbers $x$ is assigned to machine $x$.
\\
In the implementation we do not check whether a machine is assigned too many edges (Line~\ref{line:check-graph-size}), but rather show in Lemma~\ref{lem:no_overflow} that not too many edges are assigned with high probability.

\subsection{$\localPhase$}
\label{sec:MPC-localPhase}
We now discuss the implementation of \localPhase{}.
Observe that $\localPhase$ is executed locally. Therefore, the for loop at Line~\ref{line:edge-partitioning} of $\parAlg$ can also be executed locally on each machine. Thus, we only explain how to process the output of $\localPhase$.

Instead of returning the set of vertices and matched edges at Line~\ref{line:local-phase-return} of $\localPhase$, each vertex that should be returned is marked as $\discard$, and each matched edge is marked as $\matched$.
After that, we need to discard edges, whose at least one endpoint has been discarded. This can be done by broadcasting information from vertices to adjacent edges. Note that some of the discarded edges might be also marked as $\matched$.

\subsection{Putting all together}
Lines~\ref{line:vertex-partitioning}, \ref{line:check-graph-size} and~\ref{line:edge-partitioning} can be implemented as described in sections~\ref{sec:vertex-edge-partitioning} and~\ref{sec:MPC-localPhase}. Lines~\ref{line:V'-update} and~\ref{line:M-update} do not need an actual implementation, as by that point all the vertices that are not marked as $\discard$ constitute $V'$, and all the edges incident to $V \setminus V'$ will be marked as $\discard$. Similarly, all the matched edges will be marked as $\matched$ by the implementation of $\localPhase$. All the edges and vertices that are marked as $\discard$ will be ignored in further processing. After all the rounds are over, the matching consists of the edges marked as $\matched$.

Let $\threshold_\star$ be the value of $\threshold$ at Line~\ref{line:remove-high-degree-vertices}, and hence the value of $\threshold$ at the end of the last while loop iteration. Let $\threshold'$ be the value of $\threshold$ just before the last iteration, i.e. $\threshold_\star = \threshold' / 2^{\phases}$, for the corresponding ${\phases}$. Now consider the last call of $\localPhase$ at Line~\ref{line:edge-partitioning}. The last invocation has $\threshold' / (2^{\phases - 1})$ as a parameter. On the other hand, by Claim~\ref{claim:random_events} and Claim~\ref{claim:good_consequences} we know that after the last invocation of $\localPhase$ with high probability there is no vertex that has degree greater then $\tfrac{3}{4} \threshold' / (2^{\phases - 1}) < 2 \threshold_\star$. Therefore, with high probability there is no vertex that should be removed at Line~\ref{line:remove-high-degree-vertices}, and hence we do not implement that line either.

An implementation of Line~\ref{line:simulate_global_alg} is described in Section~\ref{sec:implementation-of-centralized-algorithm}. Finally, we can state the following result.
\begin{lemma}
There exists an implementation of $\parallelAlgorithm$ in the MPC model that with high probability executes $O\left((\log \log n)^2 + \max{\left(\log{\tfrac{n}{S}}, 0 \right)} \right)$ rounds.
\end{lemma}
\begin{proof}
	In the proof we analyze the case $S \leq n$. Otherwise, for the case $S > n$, we think of each machine being split into $\lfloor S/n \rfloor$ "smaller" machines, each of the smaller machines having space $n$.

	We will analyze the number of iterations of the while loop $\parAlg$ performs. Let $\threshold_i$ and $\phases_i$ be the value of $\threshold$ and $\phases$ at the end of iteration $i$, respectively. Then, from Line~\ref{line:machines-assignment} and Line~\ref{line:phases-assignment} we have
	\[
		\phases_i = \left\lceil \frac{1}{16}\log_{120\prHmult}\left(\threshold_{i - 1} /\machines\right) \right\rceil \ge \frac{1}{16}\log_{120\prHmult}\left(\threshold_{i - 1} /\machines\right) \ge \frac{1}{16}\log_{120\prHmult}\sqrt{\frac{S \threshold_{i - 1}}{n}}.
	\]
	Define $\gamma := \tfrac{1}{32 \log_{2} 120\prHmult}$. By plugging in the above bound on $\phases_i$, from Line~\ref{line:next-threshold}, we derive
	\begin{equation}\label{eq:threshold-i-recurence}
		\threshold_i = \threshold_{i-1}\cdot 2^{\phases_i} \le \threshold_{i - 1} \cdot 2^{-\frac{1}{16}\log_{120\prHmult}\sqrt{\frac{S \threshold_{i - 1}}{n}}} = \threshold_{i - 1} \cdot 2^{-\frac{\log_{2}\frac{S \threshold_{i - 1}}{n}}{32 \log_{2} 120\prHmult}} = \threshold_{i - 1}^{1 - \gamma} \left(\frac{n}{S}\right)^{\gamma}
	\end{equation}

	To obtain the number of iterations the while loop of $\parAlg$ performs, we derive for which $i \ge 1$ the condition at Line~\ref{line:par_alg:parallel-while-condition} does not hold.
	
	Unraveling $\threshold_{i - 1}$ further from~\eqref{eq:threshold-i-recurence} gives
	\begin{equation}\label{eq:threshold-i-bound}
		\threshold_i \le \threshold_0^{(1 - \gamma)^i} \left(\frac{n}{S}\right)^{\gamma \sum_{j = 0}^{i - 1}(1 - \gamma)^j} \le n^{(1 - \gamma)^i} \left(\frac{n}{S}\right)^{\gamma \frac{1 - (1 - \gamma)^i}{1 - (1 - \gamma)}} = n^{(1 - \gamma)^i} \left(\frac{n}{S}\right)^{1 - (1 - \gamma)^i}
	\end{equation}
	Observe that $(c \log \log n)^{-1} \le \gamma \le (32 \log \log n)^{-1} < 1/2$, for an absolute constant $c$ and $n \ge 4$.

	For $S \le n$ and as $\gamma < 1/2$ we have
	\begin{equation}\label{eq:n-S-bound}
		\left(\frac{n}{S}\right)^{1 - (1 - \gamma)^i} \le \frac{n}{S}.
	\end{equation}
	On the other hand, for $\istar = \tfrac{\log \log n}{\gamma} \le c (\log \log n)^2$ we have
	\begin{equation}\label{eq:n-istar-bound}
		n^{(1 - \gamma)^\istar} < \log{n}.
	\end{equation}
	Now putting together \eqref{eq:threshold-i-bound}, \eqref{eq:n-S-bound}, and \eqref{eq:n-istar-bound} we conclude
	\[
		\threshold_\istar < \frac{n}{S} \ln{n},
	\]
	and hence the number of iteration the while loop of $\parAlg$ performs is $O\left((\log \log n)^2\right)$.
			
	\paragraph{Total round complexity.} Every iteration of the while loop can be executed in $O(1)$ MPC rounds with probability at least $1 - 1/n^3$. Since there are $O\left((\log \log n)^2\right)$ iterations of the while loop, all the iterations of the loop can be performed in $O\left((\log \log n)^2\right)$ many rounds with probability at least $1 - 1/n^2$.
	
	On the other hand, by Lemma~\ref{lemma:globalAlgorithm-implementation} and the condition at Line~\ref{line:par_alg:parallel-while-condition} of $\parAlg$, the computation of Line~\ref{line:simulate_global_alg} of $\parAlg$ can be performed in $O\left(\log{\left(\tfrac{n}{S} (\ln n)^{32}\right)}\right)$ rounds. Putting the both bounds together we conclude that the round complexity of $\parAlg$ is $O\left((\log \log n)^2 + \log{\tfrac{n}{S}} \right)$ for the case $S \le n$. For the case $S > n$ (recall that in this regime we assume that each machine is divided into machines of space $n$) the round complexity is $O\left((\log \log n)^2\right)$.
\end{proof}

\subsection*{Acknowledgments}
We thank Sepehr Assadi, Mohsen Ghaffari, Cameron Musco, Christopher Musco, Seth Pettie, and Govind Ramnarayan for helpful discussions and comments.
A.~Czumaj was supported in part by the Centre for Discrete Mathematics and its Applications (DIMAP),
Royal Society International Exchanges Scheme 2013/R1, 
IBM Faculty Award, and
the EPSRC award EP/N011163/1.
A.~Mądry was supported in part by an Alfred P.~Sloan Research Fellowship, Google Research Award, and the NSF grant CCF-1553428.
S.~Mitrovi{\' c} was supported in part by the Swiss NSF grant P1ELP2\_161820. 
P.~Sankowski was partially supported by grant NCN2014/13/B/ST6/00770 of Polish National Science Center and ERC StG grant TOTAL no.\ 677651.
\bibliographystyle{alpha}
\bibliography{bibliography}


\end{document}